\documentclass[11pt]{article}
\usepackage[letterpaper,margin=1in]{geometry}
\usepackage{odonnell}

\newcommand{\proves}[1]{\vdash_{{#1}}}
\newcommand{\lnote}[1]{}
\newcommand{\rnote}[1]{}

\newcommand{\FR}[2]{\mathrm{FR}^{#1}_{#2}}
\newcommand{\minvc}{\mathsf{Min}\text{-}\mathsf{VC}}
\newcommand{\maxis}{\mathsf{Max}\text{-}\mathsf{IS}}

\begin{document}

\title{Hypercontractive inequalities via SOS, and the Frankl--R\"odl graph} 

\author{
Manuel Kauers\thanks{Institute for Algebra, Johannes Kepler Universit\"at. Supported by FWF grant Y464-N18.} \\
\texttt{manuel.kauers@jku.de}
\and
Ryan O'Donnell\thanks{Department of Computer Science, Carnegie Mellon University. Supported by NSF grants CCF-0747250 and CCF-1116594, a Sloan fellowship, and a grant from the MSR--CMU Center for Computational Thinking.}  \\
\texttt{odonnell@cs.cmu.edu}
\and
Li-Yang Tan\thanks{Department of Computer Science, Columbia  University. Research done while visiting CMU.} \\
\texttt{liyang@cs.columbia.edu}
\and
Yuan Zhou$^\dagger$\thanks{Also supported by a grant from the Simons Foundation (Award Number 252545).} \\
\texttt{yuanzhou@cs.cmu.edu}
}
\maketitle

\abstract{
    Our main result is a formulation and proof of the reverse hypercontractive inequality in the sum-of-squares (SOS) proof system.  As a consequence we show that for any constant $0 < \gamma \leq 1/4$, the SOS/Lasserre SDP hierarchy at degree $4\lceil \frac{1}{4\gamma}\rceil$ certifies the statement ``the maximum independent set in the Frankl--R\"odl graph $\FR{n}{\gamma}$ has fractional size~$o(1)$''.  Here $\FR{n}{\gamma} = (V,E)$ is the graph with $V = \{0,1\}^n$ and $(x,y) \in E$ whenever $\hamdist(x,y) = (1-\gamma)n$ (an even integer).  In particular, we show the degree-$4$ SOS algorithm certifies the chromatic number lower bound ``$\chi(\FR{n}{1/4}) = \omega(1)$'', even though $\FR{n}{1/4}$ is the canonical integrality gap instance for which standard SDP relaxations cannot even certify ``$\chi(\FR{n}{1/4}) > 3$''.
    Finally, we also give an SOS proof of (a generalization of) the sharp $(2,q)$-hypercontractive inequality for any even integer~$q$.
}

\setcounter{page}{0}
\thispagestyle{empty}
\newpage

\section{Introduction}

Hypercontractive inequalities play an important role in analysis of Boolean functions. They are concerned with the \emph{noise operator}~$T_\rho$ which acts on functions $f \btR$ via $T_\rho f(x) = \Ex[f(\by)]$, where $\by$ is a ``$\rho$-correlated copy'' of~$x$.  Equivalently, $T_\rho f = \sum_{S \subseteq [n]} \rho^{|S|} \wh{f}(S) \chi_S$, where the numbers $\wh{f}(S)$ are the Fourier coefficients of~$f$.  The standard hypercontractivity inequality was first proved by Bonami~\cite{Bon70} and the reverse hypercontractivity inequality was first proved by Borell~\cite{Bor82}.  We state both, recalling the notation $\|f\|_p = \Ex_{\bx \sim \bn}[|f(\bx)|^p]^{1/p}$.

\begin{named}{Hypercontractive Inequality} Let $f \btR$, let $1 \leq p \leq q \leq \infty$, and let $0 \leq \rho \leq \sqrt{(p-1)/(q-1)}$.  Then $\|T_\rho f\|_q \leq \|f\|_p$.
\end{named}

\begin{named}{Reverse Hypercontractive Inequality} Let $f \btR^{\geq 0}$, let $-\infty \leq q \leq p \leq 1$, and let $0 \leq \rho \leq \sqrt{(1-p)/(1-q)}$. Then $\|T_\rho f\|_q \geq \|f\|_p$.
\end{named}

The hypercontractive inequality is almost always used with either $p = 2$ or $q = 2$.  The $(2,4)$-hypercontractivity inequality --- i.e., the case $q = 4$, $p = 2$, $\rho = 1/\sqrt{3}$ --- is a particularly useful case, as is the following easy corollary:
\begin{theorem}                                     \label{thm:2-4-hypercon}
    For $k \in \N$, let $\calP^{\leq k}$ be the projection operator which maps $f \btR$ to its low-degree part $\calP^{\leq k} f = \sum_{|S| \leq k} \wh{f}(S) \chi_S$.  Then the $2\rightarrow 4$ operator norm of $\calP^{\leq k}$ is at most $3^{k/2}$.  I.e., $\|\calP^{\leq k} f\|_4 \leq 3^{k/2} \|f\|_2$.
\end{theorem}
Theorem~\ref{thm:2-4-hypercon} is known to have a proof which is noticeably simpler than that of the general hypercontractive inequality~\cite{MOO05}.  Theorem~\ref{thm:2-4-hypercon} can be used to prove, e.g., the KKL~Theorem~\cite{KKL88}, the sharp small-set expansion statement for the $1/3$-noisy hypercube, and the Invariance Principle of~\cite{MOO10}.  More generally, the hypercontractivity inequality has the following corollary:
\begin{theorem}                                     \label{thm:2-q-hypercon}
    For any $q \geq 2$ and $f \btR$ we have $\|\calP^{\leq k} f\|_q \leq (q-1)^{k/2} \|f\|_2$.
\end{theorem}
\noindent This corollary is often used to control the behavior of low-degree polynomials of random bits.

Reverse hypercontractivity is perhaps most often used to show that if $A, B \subseteq \bn$ are large sets and $(\bx, \by)$ is a $\rho$-correlated pair of random strings then there is a substantial chance that $\bx \in A$ and $\by \in B$.  This was first deduced in~\cite{MOR+06} by deriving the following consequence of reverse hypercontractivity: 
\begin{theorem}                                     \label{thm:morss}
    Let $f, g \btR^{\geq 0}$, let $0 \leq q \leq 1$, and let $0 \leq \rho \leq 1-q
    $.  Then $\E[f(\bx) g(\by)] \geq \|f\|_q \|g\|_q$ when $(\bx,\by)$ is a pair of $\rho$-correlated random strings.
\end{theorem}
The reverse hypercontractive inequality has been used, e.g., in problems related to approximability and hardness of approximation~\cite{FKO07,She09a,BHM12}, and problems in quantitative social choice~\cite{MOO10,Mos12a,MOS12b,Kel12,MR12}.

\ignore{

One of the widely used hypercontractive inequalities, the $(2,4)$-hypercontractive inequality, states
\begin{theorem}\label{thm:2-4-hypercon}
Let $\mathcal{P}_d$ be the linear operator that maps a function $f : \{\pm 1\}^n \to \R$ of the form $f = \sum_{S \subseteq [n]} \hat{f}(S) \chi_S$ to its low-degree part $f' = \sum_{|S| \leq d} \hat{f}(S) \chi_S$ (where $\chi_S(x) = \prod_{i \in S} x_i$), then the $2\to4$ operator norm of $\mathcal{P}_d$ is at most $3^{d/2}$. In other words, the operator $\mathcal{P}_d$ satisfies
\begin{align}
 \|\mathcal{P}_d f\|_4 \leq 3^{d/2} \|f\|_2, \qquad \forall f : \{\pm 1\}^{n} \to \R . \label{eq:2-4-hypercon}
\end{align}
\end{theorem}
The $(2,4)$-hypercontractive inequality lies in the heart of the proof of the KKL~theorem~\cite{KKL88} and is used in the soundness analysis of the construction of integrality gap instances for Unique Games \cite{KV05, RS09b, KPS10, BGH+12}, Balanced Separator \cite{DKSV06}. Some other applications (e.g. \cite{KKMO07, MOO10, OW09}) need the more general Bonami-Beckner inequality \cite{Bon70, Bec75}, which is stated as follows.
\begin{theorem}\label{thm:bonami-beckner}
Let $f : \{\pm 1\}^n \to \R$ and $q \geq p \geq 1$. Then
\begin{align}
\|T_{\rho} f\|_q \leq \|f \|_p, \qquad \forall \rho : 0 \leq \rho \leq (p-1)^{1/2} / (q-1)^{1/2} , \label{eq:bonami-beckner}
\end{align}
\end{theorem}
Here $T_{\rho}$ is the well-known ``noise-operator'' from analysis of Boolean functions.

Another way to generalize the $(2, 4)$-hypercontractive inequality is the $(2, 2s)$-hypercontractive inequality for integers $s \geq 1$, which upper bounds the $2\to 2s$ operator norms of $\mathcal{P}_d$ defined in Theorem~\ref{thm:2-4-hypercon}. Indeed, we can state the theorem as follows, relaxing the condition that underlying distribution is a bit.
\begin{theorem}\label{thm:2-2s-hypercon}
Assume $\bx_1, \dots, \bx_n$ are independent real random variables with mean~$0$, second moment~$1$, and higher moments satisfying $\E[\bx_i^{2k+1}] = 0$ and $\E[\bx_i^{2k}] \leq (2k-1)!!$ for all integers $k > 1$.  (E.g., Rademacher, standard Gaussian, and uniform on $[-\sqrt{3}, \sqrt{3}]$ random variables are all acceptable.)

Let $f=\sum_{|S| \leq d} \hat{f}(S) \chi_S$ (where $\chi_S(x) = \prod_{i \in S} x_i$) be an $n$-variate polynomial of degree at most $d$ and write $\bF = f(\bx_1, \dots, \bx_n)$.  Then for all integers $s \geq 1$,
\[
\|\bF\|_{2s} \leq (2s-1)^{d/2} \|\bF\|_2.
\]
\end{theorem}
Both Theorem~\ref{thm:2-4-hypercon} and Theorem~\ref{thm:2-2s-hypercon} follow from the Bonami-Beckner inequality, but the standard proof of the Bonami-Beckner inequality is a bit technical. A very easy inductive proof of Theorem~\ref{thm:2-4-hypercon} was published in \cite{MOO10}. In Appendix~\ref{app:2-2s-hypercon-proof}, we present a similar ``easy'' proof of Theorem~\ref{thm:2-2s-hypercon}, modeled on the \cite{MOO10} proof.
}

\subsection{Sum-of-squares proofs of hypercontractive inequalities}

The present work is concerned with proving hypercontractive inequalities via ``sums of squares'' (SOS); i.e., in the Positivstellensatz proof system introduced by Grigoriev and Vorobjov~\cite{GV01}.  A recent work of Barak et al.~\cite{BBH+12} showed that the Khot--Vishnoi~\cite{KV05} SDP integrality gap instances for Unique-Games are actually well-solved by the ``$4$-round Lasserre SDP hierarchy''; equivalently, the ``degree-$8$ SOS hierarchy''.  This is despite the fact that they are strong gap instances for superconstantly many rounds of other weaker SDP hierarchies such as Lov\'{a}sz--Schrijver$^+$ and Sherali--Adams$^+$~\cite{RS09b,KS09}.  The key to analyzing the optimum value of the Khot--Vishnoi instances is the hypercontractive inequality, and perhaps the key technical component of the Barak et al.\ result is showing that Theorem~\ref{thm:2-4-hypercon} has a degree-$4$ ``SOS proof''.  That is, if we treat each $f(x)$ as a formal ``indeterminate'', then $9^{k} \|f\|_2^4 - \|\calP^{\leq k} f\|_4^4$ is a degree-$4$ polynomial in the $2^n$ indeterminates, and Barak et al.\ showed that it is a sum of squared polynomials (hence always nonnegative).

The connection between SOS proofs and SDP relaxations for optimization problems was made independently by Lasserre~\cite{Las00} and Parrilo~\cite{Par00}.  Roughly speaking, if a system of $n$-variate polynomial inequalities can be refuted within the degree-$d$ SOS proof system of Grigoriev and Vorobjov~\cite{GV01}, then this refutation can also be found efficiently by solving a semidefinite program of size $n^{O(d)}$.  The associated ``degree-$d$ SOS hierarchy'' for approximating optimization problems is known to be at least as strong as the Lov\'{a}sz--Schrijver$^+$ and Sherali--Adams$^+$ SDP hierarchies, and the~\cite{BBH+12} result shows that it can be noticeably stronger for the notorious Unique Games problem. (For more details, see e.g.~\cite{OZ13,BBH+12}.)

Later,~\cite{OZ13} showed that the degree-$4$ SOS hierarchy correctly analyzes the value of the~\cite{DKSV06} instances of Balanced-Separator, which are known to be superconstant-factor integrality gap instances for superconstantly many rounds of the ``LH SDP hierarchy''~\cite{RS09b}. It was also shown in~\cite{OZ13} that the degree-$O(1)$ SOS hierarchy certifies the value of the~\cite{KV05} instances of Max-Cut to within factor~$.952$, whereas superconstantly many rounds of the Sherali--Adams$^+$ hierarchy are still off by a factor of~$.878$~\cite{RS09b,KS09}.  (The $.952$ here was very recently improved to any $1-\eps$.~\cite{DMN13})  The key to the former result was an SOS proof of the KKL~Theorem (relying on \cite{BBH+12}'s SOS proof of Theorem~\ref{thm:2-4-hypercon}); the key to the latter was an SOS proof of an Invariance Principle variant, which in turn needed an SOS proof of higher-norm hypercontractivity, Theorem~\ref{thm:2-q-hypercon}.  The work~\cite{OZ13} was unable to actually obtain Theorem~\ref{thm:2-q-hypercon} with an SOS proof, but instead obtained a weaker version which sufficed for their purposes.

Still, the full power of the SOS hierarchy is far from well-understood.  Analyzing what can and cannot be proved with low-degree SOS proofs is evidently very important; for example, it's  consistent with our current knowledge that the degree-$4$ SOS hierarchy refutes the Unique-Games Conjecture, gives a $1.01$-approximation for Uniform Sparsest-Cut,
a $1.4$-approximation for Vertex-Cover, and certifies that any graph with chromatic number exceeding~$5$ is not $3$-colorable.

In particular, hypercontractive inequalities have played a key role in many of the sophisticated SDP integrality gap instances.  Thus it is natural to ask: Can a sharp version of the hypercontractive inequality be proved in the SOS proof system?  Can any version of the reverse hypercontractive inequality be proved?  As we will see, the latter question is particularly relevant for the known SDP integrality instances of the $3$-Coloring and Vertex-Cover problems.

\ignore{
Given the lack of understanding about the SOS/Lasserre SDP, it is worthwhile to explore what we can prove in the SOS proof system. A concrete question at this point is: can we prove the $(2, 2s)$-hypercontractive inequality, or even the full Benami-Beckner inequality in the SOS proof system? Unfortunately, for most values of $p$, the $p$-``norms'' in the inequality \eqref{eq:bonami-beckner} are not polynomials in the values $f(x)$ even after we raise it to the $p$-th power. Therefore, we need first to identify a valid statement in the SOS proof system.

Since the $(2, 4)$-hypercontractive inequality plays an central role in many integrality gap constructions, the SOS proof of the inequality has several computational implications. Combining with other SOS proof techniques, \cite{BBH+12} showed that the level-$4$ SOS/Lasserre SDP algorithm refutes all the known integrality gap instances for Unique Games (for many other strong SDP relaxations) \cite{KV05, RS09b, KPS10, BGH+12}. Later, \cite{OZ13} showed that level-$2$ SOS/Lasserre SDP refutes the Sherali--Adams$+$SDP gap instance for Balanced Separator \cite{DKSV06}, where the key is to extend the low-degree SOS proof of the $(2, 4)$-hypercontractive inequality to a low-degree proof of the KKL theorem.

The SOS/Lasserre SDP, proposed independently by Lasserre~\cite{Las00} and Parrilo~\cite{Par00}, is the most powerful SDP relaxation hierarchy studied in the literature (compared to Lov\'{a}sz--Schrijver hierarchy \cite{LS91a}, Sherali--Adams LP hierarchy \cite{SA90}, and mixed hierarchies combining LP relaxation hierarchies with SDP, such as Lov\'{a}sz--Schrijver$^+$, Sherali--Adams$^+$). However, the full power of the SOS/Lasserre SDP is not well understood -- it is consistent with our current knowledge that level-2 SOS/Lasserre SDP could refute the Unique Games Conjecture.

Given the lack of understanding about the SOS/Lasserre SDP, it is worthwhile to explore what we can prove in the SOS proof system. A concrete question at this point is: can we prove the $(2, 2s)$-hypercontractive inequality, or even the full Benami-Beckner inequality in the SOS proof system? Unfortunately, for most values of $p$, the $p$-``norms'' in the inequality \eqref{eq:bonami-beckner} are not polynomials in the values $f(x)$ even after we raise it to the $p$-th power. Therefore, we need first to identify a valid statement in the SOS proof system.
}

\subsection{Our results}

The main result in this paper is an SOS proof of the reverse hypercontractivity Theorem~\ref{thm:morss} for all~$q$ equal to the reciprocal of an even integer.  As one application of this, we show that just the degree-$4$ algorithm from the SOS hierarchy can certify that the ``Frankl--R\"odl'' SDP integrality gap instances for $3$-Coloring have chromatic number~$\omega(1)$.
Finally, we also give an SOS proof of the sharp $(2,q)$-hypercontractive inequality for all even integers~$q$; in fact, a version with relaxed moment conditions.  We find it interesting to see that the two powerful hypercontractive inequalities admit proofs as ``elementary'' as sum-of-squares proofs. On the other hand, to obtain these proofs we had to use somewhat elaborate methods, including computer algebra techniques.

\paragraph{The hypercontractive inequality for even integer norms.}  As mentioned, Barak et al.~\cite{BBH+12} gave an SOS proof of Theorem~\ref{thm:2-4-hypercon}, that $\|\calP^{\leq k} f\|_4^4 \leq 9^k\|f\|_2^4$.  Although there is a very easy proof of this theorem ``in ZFC''~\cite{MOO05}, that proof uses the Cauchy--Schwarz inequality, whose square-roots do not obviously translate into SOS statements.  The SOS proof in~\cite{BBH+12} gets around this by proving the \emph{generalized} statement $\E[(\calP^{\leq k} f)^2 (\calP^{\leq k'} g)^2] \leq 3^{k+k'} \E[f^2]\E[g^2]$, allowing them to replace Cauchy--Schwarz with $XY \leq \half X^2 + \half Y^2$.  In~\cite{OZ13} this SOS proof was very slightly generalized to cover the $(2,4)$-hypercontractive inequality, $\E[(T_{\rho} f)^2(T_{\rho} g)^2] \leq \E[f^2]\E[g^2]$ for $\rho = 1/\sqrt{3}$.  That work also gave an SOS proof of a weakened version of Theorem~\ref{thm:2-q-hypercon} for all even integers~$q$, namely $\|\calP^{\leq k} f\|_q^q \leq q^{O(qk/2)} \|f\|_2^{q}$.  (Attention is restricted to even integers~$q$ because the $(2,q)$-hypercontractive inequality cannot even be stated as a polynomial inequality otherwise.)

In Section~\ref{sec:hypercon} we prove the full
$(2,q)$-hypercontractive inequality for all even integers~$q$.  Our
strategy is as follows. First, we give a simple proof
  (``in ZFC'') of $(2,q)$-hypercontractivity for all even integers
  $q$; our proof works not just for random~$\pm 1$ bits but for any
  random variables satisfying fairly liberal moment bounds.  Indeed,
  we are not aware of any previous work showing that such moment
  bounds are sufficient for hypercontractivity.  However this proof
  relies on the well-known fact that the hypercontractivity inequality
  tensorizes \cite{KS88}, which in turn uses the triangle inequality
  for the $(q/2)$-norm, an inequality that cannot even be stated in SOS.
  For our SOS extension of this result we move to a $(q/2)$-function
  version of the statement as in \cite{BBH+12}; this requires some 
  more work. 
   Our final theorem is as follows:
\begin{theorem} \label{thm:informal-hypercon} (Informal.) Let $s \in \N^+$ and write $q = 2s$.  Let $0 \leq \rho \leq \frac{1}{\sqrt{q-1}}$.  Let $\bx = (\bx_1, \dots, \bx_n)$ be a sequence of independent real random variables, with each $\bx_i$ satisfying
    \[
        \E[\bx_i^{2j-1}] = 0, \quad \E[\bx_i^{2j}] \leq (2s-1)^j \tfrac{\tbinom{s}{j}}{\tbinom{2s}{2j}}\quad \text{for all integers $1 \leq j \leq s$};
    \]
    further assume that $\E[\bx_i^2] = 1$ for each $i$. (Rademachers and standard Gaussians qualify.)
    Then for functions $f_1, \dots, f_s : \R^n \to \R$ there is an SOS proof of
    \[
        \E\left[\littleprod_{i=1}^s (T_\rho f_i(\bx))^2\right] \leq \prod_{i=1}^s \E[f_i(\bx)^2].
    \]
    As corollaries we have SOS proofs of $\|T_\rho f\|_q^q \leq \|f\|_2^{q}$ and $\|\calP^{\leq k} f\|_q^q \leq (q-1)^{qk/2}\|f\|_2^{q}$.\rnote{these things follow as in OZ13; not sure we need to bother writing it}
\end{theorem}

\paragraph{The reverse hypercontractive inequality.}  Giving an SOS proof of this theorem proved to be significantly more difficult; it is our main result and the source of our application to $3$-Coloring and Vertex-Cover integrality gaps.  The theorem cannot even be stated in the SOS proof system directly since the $p$-``norms'' are not polynomials in the values~$f(x)$ when $p < 1$.  We turn to the $2$-function version from~\cite{MOR+06}, Theorem~\ref{thm:morss}; if $q = \frac{1}{2k}$ for some $k \in \N^+$ and if we replace $f$ and $g$ by $f^{2k}$ and $g^{2k}$ then we get a polynomial statement (and we can even drop the hypothesis that~$f$ and~$g$ are nonnegative).  The resulting theorem is:
\begin{theorem}                                     \label{thm:our-reverse-informal}
    Let $k \in \N^+$ and let $0 \leq \rho \leq 1-\frac{1}{2k}$.  Then for functions $f, g \btR$ there is a degree $4k$ SOS proof of
    \[
        \E_{\substack{(\bx,\by) \\ \text{$\rho$-corr'd}}} [f(\bx)^{2k}g(\by)^{2k}] \geq \E[f]^{2k}\E[g]^{2k}.
    \]
\end{theorem}
We prove this result in Section~\ref{sec:reverse-hypercon}.  An induction on~$n$ easily reduces the problem to the $n = 1$ case; for each~$k$, this is an inequality in four real indeterminates.  Then by homogeneity we can further reduce to an inequality in just two indeterminates.  Nevertheless, giving an SOS-proof of this ``two-point inequality'' for all~$k$ seems to be surprisingly tricky.  As an example of the problem we need to solve (the $k = 3$ case), the reader is invited to try the following puzzle:

``Show that
\[
{\tfrac {11}{24}} \left( 1+a \right) ^{6} \left( 1+b
 \right) ^{6}+{\tfrac {11}{24}} \left( 1-a \right) ^{6}
 \left( 1-b \right) ^{6}+\tfrac{1}{24} \left( 1+a \right) ^{6}
 \left( 1-b \right) ^{6}+\tfrac{1}{24} \left( 1-a \right) ^{6}
 \left( 1+b \right) ^{6}-1
\]
is a sum of squared polynomials in $a$ and $b$.''

Our solution is presented in Section~\ref{sec:main}. Our high level approach is to employ a change of variables which reduces the task to proving a sequence of one-variable real inequalities.  This is helpful because every nonnegative univariate polynomial is SOS; hence we can use any mathematical technique to verify the one-variable inequalities.  We establish the one-variable inequalities using techniques from computer algebra. Peculiarly, this approach only works for the specific choice $\rho = 1 - \frac{1}{2k}$; however the proof for general $0 \leq \rho \leq 1 - \frac{1}{2k}$ can be deduced since the two-point inequality is linear in~$\rho$.

\subsubsection{Application to integrality gap instances for $3$-Coloring and Vertex-Cover} \label{sec:intro-application}

Along the lines of~\cite{BBH+12,OZ13}, our SOS proof of the reverse hypercontractive inequality also has applications to integrality gap instances; specifically, for the $3$-Coloring and Vertex-Cover problems. These problems can be put in a common framework by considering the Maximum Independent-Set problem:
\begin{definition}
    Given graph $G = (V,E)$ we define the (fractional) size of its maximum independent~set:
    \[
        \maxis(G) \ \ =\ \  \max \{|S|/|V| : S \subseteq V \text{ such that } E \cap (S \times S) = \emptyset \} \ \ \in\ \  [0,1],
    \]
\end{definition}
There is a one-way connection with $k$-Coloring: any graph~$G$ with chromatic number $\chi(G) \leq k$ has $\maxis(G) \geq 1/k$.  There is a two-way connection with the Vertex-Cover problem: a set $S \subseteq V$ is independent if and only if its complement $\overline{S} = V \setminus S$ is a vertex cover (i.e., every $e \in E$ meets $\overline{S}$).  Thus $\minvc(G)$, the minimum (fractional) size of a vertex cover in~$G$, is equal to $1 - \maxis(G)$.

Finding the chromatic number or minimum vertex cover of a graph is an $\NP$-hard problem; thus it has been common to seek efficient approximation algorithms.  For example, one may seek an efficient algorithm which can $100$-color any $3$-colorable graph, or find a vertex cover of size at most~$1.5$ times the minimum.  Neither of these problems is known to be polynomial-time solvable; nor is either known to be $\NP$-hard.  In fact, the $3$-colorability question shows an enormous gap; we only know an efficient algorithm for $n^{.2111}$-coloring $3$-colorable graphs~\cite{ACC06}, and $\NP$-hardness of $4$-coloring them~\cite{KLS00, GK04}.  For Vertex-Cover, there is an easy linear-time $2$-approximation algorithm~\cite[Gavril 1974]{GJ79}, whereas achieving a $1.36$-approximation is known to be $\NP$-hard~\cite{DS05}.

\paragraph{Previous work on integrality gaps.}  Based on the $40$-year lack of progress on the algorithms side, it is reasonable to suspect that there is no efficient $(2-\eps)$-approximation algorithm for Vertex-Cover.  Similarly, one may suspect that there is no efficient algorithm for $O(1)$-coloring $3$-colorable graphs.  Indeed, these statements are known to be true assuming the Unique-Games Conjecture in the first case~\cite{KR08}, and a closely related variant of the Unique-Games Conjecture in the second~\cite{DMR09}.  However there is reasonable doubt about the Unique-Games Conjecture~\cite{ABS10} and it's important to seek alternative evidence of hardness.  One very good form of evidence is showing that strong, generic polynomial-time optimization algorithms fail to give good approximations to the value of the optimal solution.  Specifically, one can seek \emph{integrality gaps} for the canonical hierarchies of linear programming and semidefinite relaxations of the problem. In this work we will often describe integrality gaps in more ``proof-theoretic language''. For example, instead of saying that for Vertex-Cover, the complete graph~$K_n$ is a factor-$\frac{n-1}{n/2}$ integrality gap instance for the linear program, we will say that linear programming ``fails to certify $\minvc(K_n) > 1/2$, even though $\minvc(K_n) = (n-1)/n$''.

There is a long line of work on integrality gaps for Chromatic-Number, Independent-Set, and Vertex-Cover.  Specific works on integrality gaps for Vertex-Cover include~\cite{KG98,Cha02,ABL02,ABLT06,Tou06,FO06,STT07a,STT07b,GMT08,GM08,Sch08,CMM09,Tul09,GMPT10,GM10,BCGM11} (see Georgiou's thesis~\cite{Geo10} for a recent survey), and papers on integrality gaps for $3$-Coloring include~\cite{KMS98,KG98,AK98a,Cha02,FLS04,AG11}.  Furthermore, almost any paper on the Lov\'asz $\vartheta$-Function~\cite{Lov79a} is implicitly concerned with integrality gaps for these problems.

Both for $3$-Coloring and Vertex-Cover, the integrality gap papers working with the strongest SDP relaxation employ the ``Frankl--R\"odl graphs'' as their hard instances:
\begin{definition}
    Let $n \in \N$ and let $0 \leq \gamma \leq 1$ be such that $(1-\gamma)n$ is an even integer.  The \emph{Frankl--R\"odl} graph $\FR{n}{\gamma}$ is the undirected graph on the $N = 2^n$ vertices $\{-1,1\}^n$ with edge set $\{(x,y) : \hamdist(x,y) = (1-\gamma)n\}$, where $\hamdist(\cdot,\cdot)$ denotes Hamming distance.
\end{definition}
The following theorem is essentially due to Frankl and R\"odl~\cite{FR87} (a few small details are only worked out in~\cite{GMPT10}):
\begin{theorem}     \label{thm:fr}
    There is a universal constant~$K$ such that for all $\gamma \leq 1/4$ it holds that $\maxis(\FR{n}{\gamma}) \leq n(1-\gamma^2/K)^n$.  In particular,
    \[
        \maxis(\FR{n}{\gamma}) \leq o_n(1), \quad  \chi(\FR{n}{\gamma}) = \omega_n(1), \quad        \minvc(\FR{n}{\gamma}) \geq 1 - o_n(1),
    \]
    whenever $\gamma \geq .1 \sqrt{\frac{\log n}{n}}$ and $n$ is sufficiently large.
\end{theorem}

For the problem of $3$-Coloring, integrality gap papers have focused mainly on $\FR{n}{1/4}$, the graph on $\{-1,1\}^n$ in which $(x,y)$ is an edge if and only if $\hamdist(x,y) = (3/4)n$, i.e., $\frac1n\la x,y \ra = -1/2$.  The succession of works~\cite{KMS98,KG98,Cha02} showed that $\FR{n}{1/4}$ is an integrality gap instance for successively stronger SDP relaxations of $3$-Coloring, with Charikar~\cite{Cha02} showing that the strongest of them still fails to certify $\chi(\FR{n}{1/4}) > 3$, even though in fact $\chi(\FR{n}{1/4}) \geq N^{\Omega(1)}$.  (Feige, Langberg, and Schechtman~\cite{FLS04} have one of the few works to employ a non-Frankl--R\"odl graph as an integrality gap instance; however it's for an SDP relaxation weaker than Charikar's.)  A recent work of Arora and Ge~\cite{AG11} shows that the degree-$\polylog(N)$ SOS proof system is able to certify $\chi(\FR{n}{1/4}) \geq 4$; we will give a much stronger result.

Turning to Vertex-Cover, there are factor-$(2-o(1)$) integrality gap instances for $N^{\Omega(1)}$ levels of the Sherali--Adams \emph{linear programming} hierarchy~\cite{CMM09}; this work does not use Frankl--R\"odl graphs.  However, as far as we are aware, the Frankl--R\"odl graphs are the only known factor-$(2-\eps$) integrality gap instances even for the basic SDP relaxation. By employing $\FR{n}{\gamma}$ with $\gamma$ slightly subconstant, it has recently been shown that $\Omega(\sqrt{\frac{\log N}{\log \log N}})$ levels of the Lov\'{a}sz--Schrijver$^+$ SDP hierarchy and $6$~levels of the Sherali--Adams$^+$ SDP hierarchy fail to certify $\minvc({\FR{n}{\gamma}}) > 1/2 + \omega(1)$, even though $\minvc({\FR{n}{\gamma}}) > 1 - o(1)$~\cite{BCGM11}.  Further, \cite{BCGM11}~conjectures (based on numerical evidence) that their $6$-level result can be extended to any constant number of levels.  Since the Sherali--Adams$^+$ hierarchy is stronger than the Sherali--Adams and Lov\'{a}sz--Schrijver$^+$ hierarchies, this conjecture would subsume the other two mentioned results, at least with regards to ruling out polynomial-time (constant-level) algorithms.

\paragraph{Our result.} As an application of our SOS proof for the reverse hypercontractive inequality, we are able to show that for any constant $0 < \gamma \leq 1/4$, the SOS/Lasserre hierarchy can certify $\maxis(\FR{n}{\gamma}) < o(1)$ using degree $4\lceil \frac{1}{4\gamma}\rceil$. In particular, whereas the strong SDP of~\cite{Cha02} fails to certify $\chi(\FR{n}{1/4}) > 3$, we show that the degree-$4$ SOS proof system correctly certifies $\chi(\FR{n}{1/4}) = \omega(1)$.  This improves the work of~\cite{AG11} which shows that degree-$\polylog(N)$ SOS proofs can certify $\chi(\FR{n}{1/4}) \geq 4$.

Our application to Vertex-Cover is not quite as strong. The prior work of~\cite{BCGM11} shows that for any constants $\eps, \gamma > 0$, $\Omega(\eps^2/\gamma)$ levels of  Lov\'{a}sz--Schrijver$^+$  and $6$~levels of Sherali--Adams$^+$ fail to certify $\minvc(\FR{n}{\gamma}) > 1/2 + \eps$, even though $\minvc(\FR{n}{\gamma}) > 1-o(1)$.  Our work shows that the SOS proof system \emph{does} certify $\minvc(\FR{n}{\gamma}) > 1-o(1)$ once the degree is as large as $1/\gamma$.  However, the \cite{BCGM11} result continues to hold for the subconstant value $\gamma = \Theta(\sqrt{\frac{\log n}{n}})$, and they advocate this parameter setting.  On the other hand, not only do we not obtain a constant-degree SOS certification when $\gamma = \Theta(\sqrt{\frac{\log n}{n}})$, our techniques do not work at all unless $\gamma \gg \frac{1}{\log n}$ (though this may be just for a technical reason). In fact, one may speculate that with the choice $\gamma = \Theta(\sqrt{\frac{\log n}{n}})$, the Frankl--R\"odl graphs \emph{are} factor-$(2-o(1))$ integrality gap instances for constant-degree SOS; see Section~\ref{sec:conclusions}.

To obtain our result we need to show an SOS proof for the Frankl--R\"odl Theorem.  A key ingredient in Frankl and R\"odl's original proof is the vertex isoperimetric inequality on~$\bn$, due to Harper.  The standard proof of this inequality involves a ``shifting'' argument which we do not see how to carry out with SOS.  However, it is known that inequalities of this type can also be proved using the reverse hypercontractive inequality studied in this paper.  In particular, Benabbas, Hatami, and Magen~\cite{BHM12} have very recently proven a ``density'' variation of the Frankl--R\"odl Theorem using the reverse hypercontractive inequality. We obtain the SOS proof for the Frankl--R\"odl Theorem by combining our SOS proof for the reverse hypercontractive inequality and an SOS version of the Benabbas--Hatami--Magen proof; see Section~\ref{sec:BHM}.

\section{Preliminaries}

\paragraph{The SOS proof system.}
We describe the SOS (Positivstellensatz) proof system of Grigoriev and Vorobjov~\cite{GV01} using the notation from the work~\cite{OZ13}; for more details, please see that paper.

\begin{definition}
    Let $X= (X_1, \dots, X_n)$ be indeterminates, let $q_1, \dots, q_m, r_1, \dots, r_{m'} \in \R[X]$, and let
    \[
        A = \{q_1 \geq 0, \dots, q_m \geq 0\} \cup \{r_1 = 0, \dots, r_{m'} = 0\}.
    \]
    Given $p \in \R[X]$ we say that $A$ \emph{SOS-proves} $p \geq 0$ \emph{with degree}~$k$, written
    \[
        A \quad\proves{k}\quad p \geq 0,
    \]
    whenever
    \begin{multline*}
        \exists v_1, \dots, v_{m'}  \text{ and SOS } u_0, u_1, \dots, u_m \text{ such that } \\
        p = u_0 + \sum_{i=1}^m u_i q_i + \sum_{j=1}^{m'} v_j r_j, \qquad \text{with } \deg(u_0), \deg(u_i q_i), \deg(v_j r_j) \leq k\ \forall i\in [m], j \in [m'].
    \end{multline*}
    Here we use the abbreviation ``$w \in \R[X]$ is SOS'' to mean $w = s_1^2 + \cdots + s_{t}^2$ for some $s_i \in \R[X]$.    We say that $A$ has a \emph{degree-$k$ SOS refutation} if
    \[
        A \quad\proves{k}\quad -1 \geq 0.
    \]
    Finally, when $A = \emptyset$ we will sometimes use the shorthand
    \[
        \proves{k}\quad p \geq 0,
    \]
    which simply means that $p$ is SOS and $\deg(p) \leq k$.
\end{definition}

\paragraph{Analysis of boolean functions.} Let us recall some standard notation from the field.  We write $\bx \sim \bn$ to denote that the string $\bx$ is drawn uniformly at random from $\bn$.  Given $f \btR$ we sometimes use abbreviations like $\E[f]$ for $\E_{\bx \sim \bn}[f(\bx)]$.  For $f, g \btR$ we also write $\la f, g \ra = \E[fg] = \Ex_{\bx \sim \bn}[f(\bx)g(\bx)]$.  For $-1 \leq \rho \leq 1$ we say that $(\bx, \by) \sim \bn \times \bn$ is a pair of \emph{$\rho$-correlated random strings} if the pairs $(\bx_i, \by_i)$ are independent for $i \in [n]$ and satisfy $\E[\bx_i] = \E[\by_i] = 0$ and $\E[\bx_i \by_i] = \rho$.  The operator $T_\rho$ acts on functions $f \btR$ via $T_\rho f(x) = \E[f(\by) \mid \bx = x]$, where $(\bx, \by)$ is a pair of $\rho$-correlated random strings.

These definitions extend straightforwardly to general product spaces $(\Omega^n, \pi^{\otimes n})$. We say that $(\bx, \by) \in \Omega^n \times \Omega^n$ is a pair of \emph{$\rho$-correlated random strings (under $\pi^{\otimes n})$} if $\bx \sim \pi^{\otimes n}$ and $\by \in \Omega^n$ is randomly chosen as follows: for each $i\in [n]$ independently, $\by_i = \bx_i$ with probability $\rho$ and $\by_i \sim \pi$ otherwise (with probability $1-\rho$).  The operator $T_\rho$ acts on functions $f : \Omega^n \to \R$ via $T_\rho f(x) = \E[f(\by) \mid \bx = x]$, where $(\bx, \by)$ is a pair of $\rho$-correlated random strings.

\paragraph{Simple SOS facts and lemmas.}  We will use the following facts and lemmas in our SOS proofs.  The first one, in particular, we use throughout without comment.

\medskip

\begin{lemma} \label{lemma:add} \ \vspace{-1.75em}
    \begin{gather*}
        \text{If} \quad A \ \ \proves{k}\ \  p \geq 0, \quad A' \ \ \proves{k'}\ \  p' \geq 0,\\
        \text{then} \quad A \cup A' \ \ \proves{\max(k,k')}\ \  p+p' \geq 0.
    \end{gather*}
\end{lemma}

\medskip

The following fact is a well-known consequence of the Fundamental Theorem of Algebra.
\begin{fact}\label{fact:univariate-sos}
A univariate polynomial $p(x)$ is SOS if it is nonnegative. In other words, we have
\[
\proves{\deg(p)} p(x) \geq 0,
\]
when $p(x) \geq 0$ for all $x \in \R$.
\end{fact}

It is also well known that for homogeneous polynomials, one can reduce the number of variables by~$1$ by ``dehomogenizing'' the polynomial, getting an SOS representation (if there is one), and rehomogenizing it to get an SOS representation of the original polynomial. Applying this trick to Fact~\ref{fact:univariate-sos}, we get:
\begin{fact}\label{fact:bivariate-homogeneous-sos}
    A homogeneous bivariate polynomial $p(x,y)$ is SOS if it is nonnegative.
\end{fact}

Here are some additional lemmas:
\begin{lemma}                                       \label{lem:powers}
    Let $c \geq 0$ be a constant and $X$ an indeterminate.  Then for any $k \in \N^+$,
    \[
        X \geq c \quad\proves{k}\quad X^k \geq c^k.
    \]
\end{lemma}
\begin{proof}
    This follows because
    \[
        X^k - c^k = (X-c + c)^k - c^k = \sum_{i=1}^{k} \tbinom{k}{i} c^{k-i} (X-c)^i
    \]
    and each power $(X-c)^i$ is either a square or $(X-c)$ times a square.
\end{proof}
\begin{lemma} \label{lem:super-CS} For any $k \in \N^+$ we have
\[
    \proves{2k}\quad \left(\tfrac{X+Y}{2}\right)^{2k} \leq \tfrac{X^{2k} + Y^{2k}}{2}.
\]
\end{lemma}
\begin{proof}
Since $\tfrac{X^{2k} + Y^{2k}}{2} - \left(\tfrac{X+Y}{2}\right)^{2k}$ is a degree-$2k$ homogeneous polynomial, the claim follows from Fact~\ref{fact:bivariate-homogeneous-sos}: the inequality is indeed true by convexity of $t \mapsto t^k$.
\end{proof}

\section{The hypercontractive inequality in SOS}\label{sec:hypercon}

 As a warmup, we give a simple proof (``in ZFC'') of the
  $(2,q)$-hypercontractive inequality $\|T_\rho f\|_q \leq \|f\|_q$
  for all even integers $q$, which implies
  Theorem~\ref{thm:2-q-hypercon} for all even integers $q$. As
  mentioned, we do this under a significantly weakened moment
  condition:
\paragraph{``$\bs$-Moment Conditions.''} \emph{For a real random variable $\bx_i$, the condition is that $\E[\bx_i^2] = 1$ and}
\[
    \E[\bx_i^{2j-1}] = 0, \quad \E[\bx_i^{2j}] \leq (2s-1)^j \tfrac{\tbinom{s}{j}}{\tbinom{2s}{2j}}\quad \textit{for all integers $1 \leq j \leq s$.}
\]
 Our proof will show that these moment conditions are
  sharp; none of them can be relaxed.

\begin{remark}  By converting to factorials and expanding, one verifies that
\[
    (2s-1)^j \tfrac{\tbinom{s}{j}}{\tbinom{2s}{2j}} = (2j-1)!! \cdot \prod_{i = 1}^{j-1} \frac{2s-1}{2s-(2i+1)}.
\]
It follows that for each fixed $j \in \N^+$, the quantity decreases as a function of $s$ (for $s \geq j$) to the limit $(2j-1)!!$, which is the $(2j)$th moment of a standard Gaussian.  This shows that a standard Gaussian and a uniformly random $\pm 1$ bit both satisfy all of the above moment conditions.
\end{remark}

\begin{theorem}
 Let $\bx=(\bx_1, \dots, \bx_n)$ be a sequence
  independent real random variables satisfying the $s$-Moment
  Conditions. Let $f : \R^n \to \R$,  $s\in\N^+$, and
  $0\leq \rho \leq \sqrt{1/(2s-1)}$. Then $\|T_\rho f(\bx)\|_{2s} \leq
  \| f(\bx)\|_2$.
\end{theorem}
\begin{proof}
  It is well-known that the hypercontractive inequality tensorizes
  \cite{KS88} (i.e.~the $n=1$ case extends to the general $n$ case by an easy induction) and so it suffices to treat the case $n=1$.  By homogeneity we may also assume $\E[f] = 1$; we thus write   $f(\bx_1) = 1 + \eps \bx_1$ for some $\eps \in \R$.  Raising both sides of the inequality
  to the $(2s)^{\text{th}}$ power and using the odd moment conditions
  ($\E[\bx_1^{2j-1}] = 0$ for all integers $1\leq j\leq s$), we have
\begin{eqnarray}
\|T_\rho f(\bx_1)\|_{2s}^{2s} & = & \sum_{j=0}^s {2s\choose {2j}}
\rho^{2j}\eps^{2j}  \E[x_1^{2j}] \label{eq:2q-LHS}\\
\|f(\bx_1)\|_2^{2s} &=& \sum_{j=0}^s {s\choose
  j}\eps^{2j}. \label{eq:2q-RHS}
\end{eqnarray}
By the even moment conditions $\E[\bx_1^{2j}] \leq (2s-1)^j{s\choose
  j}/{{2s}\choose {2j}}$, each summand in (\ref{eq:2q-LHS}) is at most
the corresponding term in (\ref{eq:2q-RHS}) and the proof is complete. 
\end{proof}

 By considering $\eps \to 0$ in (\ref{eq:2q-LHS}) and
  (\ref{eq:2q-RHS})
it is easy to see for each $j = 1, 2, \dots, s$ in turn that the associated $s$-moment condition cannot be further relaxed.

Our SOS extension
of this result requires the following lemma:
\begin{lemma}
\label{lem:hypercon-lem}
Let $v$ be an even positive integer and let $G_1, \dots, G_v, H_1, \dots, H_v$ be indeterminates.  Then
\[
    \proves{2v}\quad \prod_{i =1}^v G_i H_i \leq \frac{1}{\tbinom{v}{v/2}}\sum_{\substack{T \subset [v] \\ |T| = v/2}} \prod_{i \in T} G_i^2 \prod_{i \in [v] \setminus T} H_i^2.
\]
\end{lemma}
\begin{proof}
    The non-SOS proof would be to just apply the AM-GM inequality.  For the SOS proof we first trivially write
    \[
        \prod_{i \in [v]} G_i H_i = \frac{1}{\tbinom{v}{v/2}} \sum_{\substack{T \subseteq V \\ |T| = v/2}} \left(\prod_{i \in T} G_i \prod_{i \in [v] \setminus T} H_i\right)\left(\prod_{i \in [v]\setminus T} G_i \prod_{i \in T} H_i\right).
    \]
    We then apply the fact that $\proves{2} XY \leq \half X^2 + \half Y^2$ to each summand to deduce
    \begin{align*}
            \proves{2v}\quad \prod_{i =1}^v G_i H_i &\leq \frac{1}{2\tbinom{v}{v/2}} \sum_{\substack{T \subseteq [v] \\ |T| = v/2}} \prod_{i \in T} G_i^2 \prod_{i \in [v] \setminus T} H_i^2 + \frac{1}{2\tbinom{v}{v/2}} \sum_{\substack{T \subseteq [v] \\ |T| = v/2}} \prod_{i \in [v] \setminus T} G_i^2 \prod_{i \in T} H_i^2\\
            &= \frac{1}{\tbinom{v}{v/2}}\sum_{\substack{T \subset [v] \\ |T| = v/2}} \prod_{i \in T} G_i^2 \prod_{i \in [v] \setminus T} H_i^2. \qedhere
    \end{align*}
\end{proof}

We are now ready to state and prove the full version of Theorem~\ref{thm:informal-hypercon}.
\begin{theorem}
    Fix $s \in \N^+$ and write $q = 2s$.  Let $0 \leq \rho \leq \frac{1}{\sqrt{q-1}}$.  Let $n \in \N$ and for each $1 \leq i \leq s$ and each $S \subseteq [n]$, introduce an indeterminate $\wh{f_i}(S)$. For each $x = (x_1, \dots, x_n) \in \R^n$ we write
    \[
        f_i(x) = \sum_{S \subseteq [n]} \wh{f_i}(S)\prod_{j \in S} x_i, \quad T_{\rho} f_i(x) = \sum_{S \subseteq [n]} \rho^{|S|}\wh{f_i}(S)\prod_{j \in S} x_i.
    \]
    Let $\bx = (\bx_1, \dots, \bx_n)$ be a sequence of independent real random variables satisfying the $s$-Moment Conditions.  Then\rnote{are we skipping the 1-function and low-degree corollaries?}
    \begin{equation}  \label{eqn:induct}
        \proves{q}\quad  \E\left[\littleprod_{i=1}^s (T_\rho f_i(\bx))^2\right] \leq \prod_{i=1}^s \E[f_i(\bx)^2].
    \end{equation}
\end{theorem}
\begin{proof}
We prove~\eqref{eqn:induct} by induction on~$n$.  The base case, $n = 0$, is trivial.  For general $n \geq 1$, we can decompose each $f_i(x)$ as
\[
    f_i(x_1, \dots, x_n) = x_n g_i(x_1, \dots, x_{n-1}) + h_i(x_1, \dots, x_{n-1}).
\]
Formally, this means introducing the shorthand $h_i(x_1, \dots, x_{n-1}) = \sum_{S \not \ni n} \wh{f_i}(S) \prod_{j \in S} x_i$, and similarly for~$g_i$.  We also introduce the notation $\bF_i = f_i(\bx)$, $\wt{\bF}_i = T_{\rho} f_i(\bx)$ for each $i$, and similarly $\bG_i, \wt{\bG}_i, \bH_i, \wt{\bH}_i$.  Note that these latter four do not depend on~$\bx_n$.  By definition we have $\wt{\bF}_i = \rho \bx_n \wt{\bG}_i + \wt{\bH}_i$.

Using the fact that $\bx_n$ is independent of all $\bG_i$, $\bH_i$ and has zero odd moments, the left-hand side of~\eqref{eqn:induct} can be written as follows:
\begin{align}
&\phantom{=}\ \E\left[\littleprod_{i=1}^s \left(\rho^2\bx_n^2 \wt{\bG}_i^2 + 2\rho\bx_n \wt{\bG}_i \wt{\bH}_i +
\wt{\bH}_i^2\right)\right]  \nonumber\\
&= \sum_{\substack{\text{partitions } \\ (U,V,W) \text{ of } [s]}} \rho^{2|U|+|V|} 2^{|V|} \E\left[\bx_n^{2|U| + |V|}\right] \E\left[\littleprod_{i \in U} \wt{\bG}_i^2 \littleprod_{i \in V} \wt{\bG}_i \wt{\bH}_i \littleprod_{i \in W} \wt{\bH}_i^2\right]  \nonumber\\
 &= \sum_{u=0}^s \sum_{\substack{v = 0\\v\text{ even}}}^{s-u} \rho^{2u+v} 2^{v} \E[\bx_n^{2u+v}] \sum_{\substack{(U,V,W) \\ |U| = u \\ |V| = v}}  \E\left[\littleprod_{i \in U} \wt{\bG}_i^2 \littleprod_{i \in W} \wt{\bH}_i^2 \littleprod_{i \in V} \wt{\bG}_i \wt{\bH}_i\right].  \label{eqn:insanity}
\end{align}
We apply Lemma~\ref{lem:hypercon-lem} to each $\littleprod_{i \in V} \wt{\bG}_i \wt{\bH}_i$ (notice that each is multiplied against an SOS polynomial) to obtain
\begin{align}
    \proves{q}\quad \eqref{eqn:insanity} &\leq \sum_{u=0}^s \sum_{\substack{v =
        0\\v\text{ even}}}^{s-u} \frac{\rho^{2u+v}2^{v}}{\tbinom{v}{v/2}}
    \E[\bx_n^{2u+v}] \sum_{\substack{(U,V,W) \\ |U| = u \\ |V| = v}}
    \sum_{\substack{T \subseteq V \\ |T| = v/2}}
    \E\left[\littleprod_{i \in U \cup T} \wt{\bG}_i^2 \littleprod_{i \in W
        \cup (V \setminus T)} \wt{\bH}_i^2\right] \nonumber\\
    & \leq  \sum_{u=0}^s \sum_{\substack{v = 0\\v\text{ even}}}^{s-u}
    \frac{2^{v}}{\tbinom{v}{v/2}}
    \frac{\tbinom{s}{u+v/2}}{\tbinom{2s}{2u+v}} \sum_{\substack{(U,V,W)
        \\ |U| = u \\ |V| = v}} \sum_{\substack{T \subseteq V \\ |T| =
        v/2}} \E\left[\littleprod_{i \in U \cup T} \wt{\bG}_i^2
      \littleprod_{i \in W \cup (V \setminus T)} \wt{\bH}_i^2\right] \nonumber\\
    & \leq   \sum_{u=0}^s \sum_{\substack{v = 0\\v\text{ even}}}^{s-u} \frac{2^{v}}{\tbinom{v}{v/2}} \frac{\tbinom{s}{u+v/2}}{\tbinom{2s}{2u+v}} \sum_{\substack{(U,V,W) \\ |U| = u \\ |V| = v}} \sum_{\substack{T \subseteq V \\ |T| = v/2}} \prod_{i \in U \cup T} \E[\bG_i^2] \prod_{i \in W \cup (V \setminus T)}\E[\bH_i^2], \label{eqn:insanity2}
\end{align}
where the second inequality uses the $s$-Moments Condition and the bound on~$\rho$, and the third inequality uses the induction hypothesis.  (Again, note that each inequality is multiplied against an SOS polynomial.) It is easy to check that $\E[\bF_i^2] = \E[\bG_i^2] + \E[\bH_i^2]$ and so  the right-hand side of~\eqref{eqn:induct} is simply
\[
    \sum_{R \subseteq [s]} \prod_{i \in R} \E[\bG_i^2] \prod_{i \in [s] \setminus R} \E[\bH_i^2].
\]
Thus to complete the inductive proof, it suffices to show that for each $R \subseteq [s]$, the coefficient on $\prod_{i \in R} \E[\bG_i^2] \prod_{i \in [s] \setminus R} \E[\bH_i^2]$ in~\eqref{eqn:insanity2} is equal to~$1$.  By symmetry, and taking the sum over $v$ first in~\eqref{eqn:insanity2}, it suffices to check that for each $r = |R| = |U \cup T| \in \{0, 1, \dots, s\}$ we have
\begin{equation} \label{eqn:use-z}
    \sum_{v' = 0}^r \frac{2^{2v'}}{\tbinom{2v'}{v'}}\frac{\tbinom{s}{r}}{\tbinom{2s}{2r}}\tbinom{r}{v'} \tbinom{s-r}{v'} = 1. 
\end{equation}
With a modest amount of work it is possible to prove this identity by ``traditional'' enumerative combinatorics methods; however it is much more efficient to simply use Zeilberger's algorithm~\cite{Zei90,PWZ97}. This algorithm automatically generates the key rational function
\[
  R(r, v') = \frac{(1+2r-s)v'(2v'-1)}{(2r-2s+1)(r-s)(1+r-v')}.
\]
Then, writing $t(r, v')$ for the expression in the sum on the
left-hand side of~\eqref{eqn:use-z}, we have
\[
 t(r+1, v') - t(r,v') = R(r,v'+1)t(r,v'+1) - R(r,v')t(r,v'),
\]
as can be verified by a trivial calculation.
Summing the above equation for $v'=0,\dots,r$ shows that
$T(r+1)-T(r)=0$, where $T(r)=\sum_{v'=0}^r t(r,v')$.
Together with the initial value $T(0)=1$, it follows
by induction that $T(r)=1$ for all~$r$, as required.
\end{proof}

\section{The reverse hypercontractive inequality in SOS}                \label{sec:reverse-hypercon}

This section is devoted to providing a proof Theorem~\ref{thm:our-reverse-informal}, the reverse hypercontractivity in the SOS proof system.  More precisely:
\begin{theorem}                                     \label{thm:our-SOS-reverse}
    Let $k, n\in \N^+$, let $0 \leq \rho \leq 1 - \frac{1}{2k}$, and let $f(x)$, $g(x)$ be indeterminates for each $x \in \bn$.  Then
\[
        \proves{4k}\quad \Ex_{\substack{(\bx, \by) \\ \text{$\rho$-corr'd}}}[f(\bx)^{2k}g(\by)^{2k}] \geq \E[f]^{2k}\E[g]^{2k}.
\]
\end{theorem}

For each fixed $k$, we prove Theorem~\ref{thm:our-SOS-reverse} by
induction on~$n$.  The $n = 1$ base case of the induction is the
following $4$-variable inequality:
\begin{theorem}
\label{thm:reverse-hypercon-base}
    Let $k \in \N^+$ and let $0 \leq \rho \leq 1 - \frac{1}{2k}$. Let $F_0, F_1, G_0, G_1$ be real indeterminates.  Then
    \[
        \proves{4k} \quad
          (\tfrac14 + \tfrac14 \rho) \Bigl(F_0^{2k} G_0^{2k} + F_1^{2k} G_1^{2k}\Bigr)
        + (\tfrac14 - \tfrac14 \rho) \Bigl(F_0^{2k} G_1^{2k} + F_1^{2k} G_0^{2k}\Bigr)
         \geq
        \Bigl(\tfrac{F_0+F_1}{2}\Bigr)^{2k}\Bigl(\tfrac{G_0+G_1}{2}\Bigr)^{2k}.
    \]
\end{theorem}

\noindent Proving this base case will be the key challenge; for now, we give the induction which proves Theorem~\ref{thm:our-SOS-reverse}.

\medskip

\begin{proof}[Proof of Theorem~\ref{thm:our-SOS-reverse}]
  Let $n > 1$.  Given indeterminates $f(x)$, $g(x)$ for $x \in \bn$, let $f_0(x)$ be shorthand for $f(x_1, \dots, x_{n-1}, 1)$, let $f_1(x)$ be shorthand for  $f(x_1, \dots, x_{n-1}, -1)$, and similarly define shorthands $g_0$, $g_1$.  Now
\begin{align*}\smash{\E_{\substack{(\bx,\by) \\ \text{$\rho$-corr'd}}}}[f(\bx)^{2k}g(\by)^{2k}] = \ &(\tfrac{1}{4}+\tfrac{1}{4}\rho)
\E[f_0(\bx)^{2k}g_0(\by)^{2k}] \\
 + \ &(\tfrac{1}{4}+\tfrac{1}{4}\rho)
\E[f_1(\bx)^{2k}g_1(\by)^{2k}] \\
 + \ &(\tfrac{1}{4}-\tfrac{1}{4}\rho)
\E[f_0(\bx)^{2k}g_1(\by)^{2k}] \\
 + \ &(\tfrac{1}{4}-\tfrac{1}{4}\rho)
\E[f_1(\bx)^{2k}g_0(\by)^{2k}].
\end{align*}
By four applications of induction, we deduce
\begin{align*}
\proves{4k}\quad \smash{\E_{\substack{(\bx,\by) \\ \text{$\rho$-corr'd}}}}[f(\bx)^{2k}g(\by)^{2k}] \geq \ &(\tfrac{1}{4}+\tfrac{1}{4}\rho)
\E[f_0(\bx)]^{2k}\E[g_0(\by)]^{2k} \\
 + \ &(\tfrac{1}{4}+\tfrac{1}{4}\rho)
\E[f_1(\bx)]^{2k}\E[g_1(\by)]^{2k} \\
 + \ &(\tfrac{1}{4}-\tfrac{1}{4}\rho)
\E[f_0(\bx)]^{2k}\E[g_1(\by)]^{2k} \\
 + \ &(\tfrac{1}{4}-\tfrac{1}{4}\rho)
\E[f_1(\bx)]^{2k}\E[g_0(\by)]^{2k}.
\end{align*}
Now applying the $n = 1$ base case of the induction (Theorem \ref{thm:reverse-hypercon-base}) to the right-hand side of the above  we conclude that
\[
    \proves{4k}\quad \E_{\substack{(\bx,\by) \\ \text{$\rho$-corr'd}}}[f(\bx)^{2k}g(\by)^{2k}] \geq
\left(\tfrac{\E[f_0(\bx)]+\E[f_1(\bx)]}{2}\right)^{2k}
\left(\tfrac{\E[g_0(\by)]+\E[g_1(\by)]}{2}\right)^{2k} = \E[f]^{2k}\E[g]^{2k}. \qedhere
\]
\end{proof}

Our remaining task is to prove the $4$-variable base case, Theorem~\ref{thm:reverse-hypercon-base}.  Let us make a few simplifications.  First, we claim it suffices to prove it in the case $\rho = \rho^* = 1-\frac{1}{2k}$.  To see this, note that
\[
    (\tfrac14 + \tfrac14 \rho) \Bigl(F_0^{2k} G_0^{2k} + F_1^{2k} G_1^{2k}\Bigr)
  + (\tfrac14 - \tfrac14 \rho) \Bigl(F_0^{2k} G_1^{2k} + F_1^{2k} G_0^{2k}\Bigr)
  - \Bigl(\tfrac{F_0+F_1}{2}\Bigr)^{2k}\Bigl(\tfrac{G_0+G_1}{2}\Bigr)^{2k}
\]
is linear in~$\rho$.  Thus if we can show it is SOS for both $\rho = 0$ and $\rho = \rho^*$, it follows easily that it is SOS for all $0 < \rho < \rho^*$.  And for $\rho = 0$ the task is easy:
\[
    \proves{4k}\quad \tfrac14\left(F_0^{2k} G_0^{2k} + F_1^{2k} G_1^{2k}\right)
  + \tfrac14\left(F_0^{2k} G_1^{2k} + F_1^{2k} G_0^{2k}\right) = \left(\tfrac{F_0^{2k} + F_1^{2k}}{2}\right)\left(\tfrac{G_0^{2k} + G_1^{2k}}{2}\right) \geq \left(\tfrac{F_0+F_1}{2}\right)^{2k}\left(\tfrac{G_0+G_1}{2}\right)^{2k}
\]
by Lemma~\ref{lem:super-CS}.
\ignore{\begin{proof}[Proof of \eqref{eqn:rev-goal2} from Lemma~\ref{lem:rev-goal3} and Lemma~\ref{lem:rev-goal4}]
We have the following polynomial identity (in the values of $a$ and $b$).
\begin{align*}
\text{(LHS of \eqref{eqn:rev-goal2})} = \frac{\rho^* - \rho}{\rho^*} \cdot \text{(LHS of \eqref{eqn:rev-goal3})} + \frac{\rho}{\rho^*} \cdot \text{(LHS of \eqref{eqn:rev-goal4})} .
\end{align*}
Since $\frac{\rho^* - \rho}{\rho^*}$ and $\frac{\rho}{\rho^*}$ are nonnegative when $0 \leq \rho \leq \rho^*$, we can prove the nonnegativity of both $\frac{\rho^* - \rho}{\rho^*} \cdot \text{(LHS of \eqref{eqn:rev-goal3})}$ and $\frac{\rho}{\rho^*} \cdot \text{(LHS of \eqref{eqn:rev-goal4})}$ via degree-$4k$ SOS proofs. Therefore \eqref{eqn:rev-goal2} also admits a degree-$4k$ SOS proof.
\end{proof}

To prove Lemma~\ref{lem:rev-goal3}, observe that by Lemma~\ref{lem:super-CS}, we have
\begin{align*}
\proves{2k} \quad \half(1+a)^{2k} + \half(1-a)^{2k} \geq 1 \qquad \text{and} \qquad \quad \proves{2k} \qquad \half(1+a)^{2k} + \half(1-a)^{2k} \geq 1 .
\end{align*}
Therefore
\begin{align*}
\proves{4k} \quad \Bigl(\half(1+a)^{2k} + \half(1-a)^{2k}\Bigr)  \Bigl(\half(1+b)^{2k} + \half(1-b)^{2k}\Bigr) \geq 1 .
\end{align*}
Now Lemma~\ref{lem:rev-goal3} follows because of the following polynomial identity (in the values of $a$ and $b$),
\begin{multline*}
  \tfrac{1}{4} \Bigl((1+a)^{2k} (1+b)^{2k} + (1-a)^{2k} (1-b)^{2k} + (1+a)^{2k} (1-b)^{2k} + (1-a)^{2k} (1+b)^{2k}\Bigr) \\
= \Bigl(\half(1+a)^{2k} + \half(1-a)^{2k}\Bigr)  \Bigl(\half(1+b)^{2k} + \half(1-b)^{2k}\Bigr).
\end{multline*}
}
Next, for clarity we make a change of variables; our task becomes showing that for real indeterminates $\mu, \nu, \alpha, \beta$,
\begin{align}
    \proves{4k} \quad
        &(\tfrac14 + \tfrac14 \rho^*) \Bigl((\mu+\alpha)^{2k} (\nu+\beta)^{2k} + (\mu-\alpha)^{2k} (\nu-\beta)^{2k}\Bigr) \nonumber\\
      +\ &(\tfrac14 - \tfrac14 \rho^*) \Bigl((\mu+\alpha)^{2k} (\nu-\beta)^{2k} + (\mu-\alpha)^{2k} (\nu+\beta)^{2k}\Bigr) - \mu^{2k}\nu^{2k} \geq 0.  \label{eqn:rev-goal1}
\end{align}

Finally, by homogeneity we can reduce the above to proving the following ``two-point inequality'':

\begin{named}{Two-Point Inequality}  Let $k \in \N^+$ and let $\rho^* = 1 - \frac{1}{2k}$.  Then
\begin{align*}
    \proves{4k} \quad P_k(a,b) \coloneqq\
        &(\tfrac14 + \tfrac14 \rho^*) \Bigl((1+a)^{2k} (1+b)^{2k} + (1-a)^{2k} (1-b)^{2k}\Bigr) \nonumber\\
      +\ &(\tfrac14 - \tfrac14 \rho^*) \Bigl((1+a)^{2k} (1-b)^{2k} + (1-a)^{2k} (1+b)^{2k}\Bigr) - 1 \geq 0.
\end{align*}
\end{named}

\begin{proof}[Proof that~\eqref{eqn:rev-goal1} follows from the Two-Point Inequality]
    Suppose we show that $P_k(a,b)$ is equal to a sum of squares, say $\sum_{i=1}^m R_i(a,b)^2$ where each $R_i(a,b)$ is a bivariate polynomial.  Viewing this as an SOS identity in~$a$ only, we deduce that $\deg_a(R_i) \leq k$  for each~$i$ since $\deg_a(P_k) \le 2k$ (here $\deg_a(R_i)$ denotes the degree of $R_i$ viewed as a univariate polynomial in $a$, and likewise $\deg_a(P_k)$); similarly, $\deg_b(R_i) \leq k$ for each~$i$. Then
    \[
        \sum_{i=1}^m (\mu^k \nu^k R_i(\tfrac{\alpha}{\mu}, \tfrac{\beta}{\nu}))^2 = \mu^{2k} \nu^{2k} \sum_{i=1}^m R_i(\tfrac{\alpha}{\mu}, \tfrac{\beta}{\nu})^2  = \mathrm{LHS}\eqref{eqn:rev-goal1},
    \]
    and in the summation each expression $\mu^k \nu^k R_i(\tfrac{\alpha}{\mu}, \tfrac{\beta}{\nu})$ is a polynomial in $\mu, \nu, \alpha, \beta$.
\end{proof}

It remains to establish the Two-Point Inequality via an SOS proof.
\begin{remark}\label{rem:rev-goal4}
    We remind the reader that there is of course a ``ZFC'' proof of the Two-Point Inequality, since it follows as a special case of the reverse hypercontractive inequality.
\end{remark}

\subsection{The Two-Point Inequality in SOS} \label{sec:main}

This section is devoted to proving the Two-Point Inequality; i.e., showing $P_k(a,b)$ is SOS. After significant trial and error,\rnote{is this okay?  I want to get across the point that we suffered :)} we were led to the crucial idea of rewriting it under the following substitutions:
\begin{gather*}
    r = a-b, \qquad
    s = a+b, \qquad
    t = ab.
\end{gather*}
We may then express
\begin{align*}
    P_k(a,b) &= -1+\left(\tfrac14 + \tfrac14 \rho^*\right) \bigl((1+t+s)^{2k} + (1+t-s)^{2k}\bigr) + \left(\tfrac14 - \tfrac14 \rho^*\right)  \bigl((1-t+r)^{2k} + (1-t-r)^{2k}\bigr) \\
    &= -1 +  \left(\tfrac12 + \tfrac12 \rho^*\right) \sum_{i=0}^k \tbinom{2k}{2i} (1+t)^{2k-2i} s^{2j} + \left(\tfrac12 -\tfrac12 \rho^*\right) \sum_{j=0}^k \tbinom{2k}{2j} (1-t)^{2k-2j} r^{2j},
\end{align*}
where we used the identity
\[
    \half\bigl((c+d)^{2k} + (c-d)^{2k}\bigr) = \sum_{i=0}^{k} \tbinom{2k}{2i} c^{2k-2i} d^{2i}.
\]
Next we use $r^2 = s^2 - 4t$ to eliminate~$r$, obtaining
\[
   P_k(a,b) = -1 + \left(\tfrac12 + \tfrac12 \rho^*\right) \sum_{i=0}^k \tbinom{2k}{2i} (1+t)^{2k-2i} s^{2i} + \left(\tfrac12 - \tfrac12 \rho^*\right) \sum_{j=0}^k \tbinom{2k}{2j} (1-t)^{2k-2j} (s^2-4t)^{j}.
\]
Now we expand $(s^2-4t)^{j}$ in the latter sum so that we can write it as an even polynomial in~$s$. We get
\begin{align*}
    \sum_{j=0}^k \tbinom{2k}{2j} (1-t)^{2k-2j} (s^2-4t)^{j}  &= \sum_{j=0}^k \tbinom{2k}{2j} (1-t)^{2k-2j} \sum_{i = 0}^j \tbinom{j}{i}s^{2i}(-4t)^{j-i} \\
    &= \sum_{i=0}^k s^{2i} \sum_{j=i}^k \tbinom{2k}{2j} (1-t)^{2k-2j} \tbinom{j}{i}(-4t)^{j-i}.
\end{align*}
Thus we have
\begin{align}
    P_k(a,b) &= -1 + \sum_{i=0}^k \left(\left(\tfrac12 + \tfrac12 \rho^*\right) \tbinom{2k}{2i}(1+t)^{2k-2i} + \left(\tfrac12 - \tfrac12 \rho^*\right)  \sum_{j=i}^k \tbinom{2k}{2j} (1-t)^{2k-2j} \tbinom{j}{i}(-4t)^{j-i} \right) s^{2i} \nonumber\\
    &= Q_{k,0}(t)+ Q_{k,1}(t)s^2 + Q_{k,2}(t)s^4 + \cdots + Q_{k,k}(t)s^{2k}, \label{eqn:Qi-sum}
\end{align}
where
\begin{gather*}
    Q_{k,0}(t) = -1 + \left(\tfrac12 + \tfrac12 \rho^*\right) (1+t)^{2k} + \left(\tfrac12 - \tfrac12 \rho^*\right) \sum_{j=0}^k \tbinom{2k}{2j} (1-t)^{2k-2j}(-4t)^j, \\
    Q_{k,i}(t) = \left(\tfrac12 + \tfrac12 \rho^*\right) \tbinom{2k}{2i}(1+t)^{2k-2i} + \left(\tfrac12 - \tfrac12 \rho^*\right)  \sum_{j=i}^k \tbinom{2k}{2j} (1-t)^{2k-2j} \tbinom{j}{i}(-4t)^{j-i}, \qquad i = 1 \dots k.
\end{gather*}

Suppose we could show that~$Q_{k,0}(t)$ and also $Q_{k,1}(t), \dots, Q_{k,k}(t)$ are nonnegative for all $t \in \R$.  Then by Fact~\ref{fact:univariate-sos} they are also SOS, and hence $P_k(a,b)$ is SOS in light of~\eqref{eqn:Qi-sum}.  This would complete the proof of the Two-Point Inequality.

In fact that is precisely what we show below, using some computer algebra assistance.  We remark, though, that is not a~priori clear that this strategy should work; i.e., that $Q_{k,0}(t), \dots, Q_{k,k}(t)$ should be nonnegative.  It does not follow from the truth of the Two-Point Inequality. To see this, observe that whereas the Two-Point Inequality is known to hold for any $0 \leq \rho \leq \rho^*$, it is \emph{not} true that $Q_{k,0}(t) \geq 0$ for all $0 \leq \rho \leq \rho^*$.\rnote{this statement doesn't actually make sense since $Q_{k,0}$ does not involve $\rho$, but I think the reader will understand what we mean}  In fact, for $k = 1$ we have
\begin{equation} \label{eqn:Q10}
    Q_{1,0}(t) = t^2 -(2-4\rho^*)t
\end{equation}
which is nonnegative for all~$t$ \emph{only} for the specific choice $\rho^* = 1-\frac{1}{2k} = \frac12$.

Nevertheless, we now complete the proof of the Two-Point Inequality by showing that $Q_{k,0}(t), \dots, Q_{k,k}(t)$ are all nonnegative.
\begin{proposition}   \label{prop:Q0}
    For each $k \in \N^+$ with $\rho^* = 1-\frac{1}{2k}$, the polynomial $Q_{k,0}(t)$ is nonnegative.
\end{proposition}
\begin{proof}
  For $k = 1$ we have $Q_{1,0}(t) = t^2$ (as noted in~\eqref{eqn:Q10}); henceforth we may assume $k \geq 2$.   For $t < 0$ we substitute $a = \sqrt{-t}$, $b = -\sqrt{-t}$ into~\eqref{eqn:Qi-sum}; since $s = a+b = 0$ we get $P_k(\sqrt{-t},-\sqrt{-t}) = Q_{k,0}(t)$. By Remark~\ref{rem:rev-goal4} we have $P_k(\sqrt{-t},-\sqrt{-t}) \geq 0$ and hence $Q_{k,0}(t) \geq 0$ for all $t < 0$.

  For $t\geq0$ we first rewrite
  \[
    Q_{k,0}(t) = {-1} + (1+t)^{2k}\left(1 - \frac{1}{4k} + \frac{1}{4k} \sum_{j=0}^k\tbinom{2k}{2j}\frac{(1-t)^{2k-2j}}{(1+t)^{2k\phantom{-2j}}}(-4t)^j\right).
  \]
  Denoting the sum in this expression by~$S_k(t)$, Zeilberger's algorithm~\cite{Zei90,PWZ97} finds the recurrence equation
  \[
    (t+1)^2 S_{k+2}(t) - 2 (t^2-6 t+1) S_{k+1}(t)+(t+1)^2 S_k (t) = 0,
  \]
  valid for all $k\geq0$.
  The recurrence can be obtained for example by simply typing the command
  \begin{center}
    \verb|SumTools[Hypergeometric][ZeilbergerRecurrence](|\\
    \verb|binomial(2*k,2*j)*(1-t)^(2*k-2*j)/(1+t)^(2*k)*(-4*t)^j, |\\
    \verb|k, j, 0..k);|
  \end{center}
  into the computer algebra software Maple. 

  Since the coefficients in this recurrence do not depend on~$k$ but only on~$t$, the recurrence can be solved in closed form. Together with the initial values $S_0(t)=1$ and $S_1(t)=\tfrac{t^2-6t+1}{(t+1)^2}$, it follows that
  $S_k(t)=\cos(4k\arctan(\sqrt t))$. (Not every computer algebra system may deliver the solution in this form; however,
  for the correctness of the proof it is sufficient to check that $\cos(4k\arctan(\sqrt t))$ is indeed a solution
  of the recurrence.  This is easy to verify.) 
  Hence,
  \begin{align*}
    Q_{k,0}(t) &= {-1} + (1+t)^{2k}\left(1 - \tfrac{1}{4k} + \tfrac{1}{4k} \cos(4k\arctan(\sqrt{t}))\right)\\
               &\geq {-1} + (1+2kt)\left(1 - \tfrac{1}{4k} + \tfrac{1}{4k} \cos(4k\arctan(\sqrt{t}))\right),
  \end{align*}
  using the fact that the parenthesized expression is clearly nonnegative.    We now split into two cases.

  \paragraph{Case 1:} $t \geq \frac{1}{2k(2k-1)}$.  In this case we simply use that $\cos(4k\arctan(\sqrt{t})) \geq -1$ to obtain
  \[
       Q_{k,0}(t) \geq  {-1} + (1+2kt)\left(1 - \tfrac{1}{2k}\right) = -\tfrac{1}{2k} +(2k-1)t,
  \]
  which is indeed nonnegative when $t \geq \frac{1}{2k(2k-1)}$.

  \paragraph{Case 2:} $0 \leq t \leq \frac{1}{2k(2k-1)}$.  In this case we use the following estimates:
    \[
    \arctan(\sqrt t) \leq \sqrt{t} \quad \forall t \geq 0, \qquad
    \cos(x) \geq \kappa(x) \coloneqq 1 - \tfrac 12 x^2 + \tfrac1{24} x^4 - \tfrac1{720} x^6 \quad \forall x \in \R.
    \]
    Note that $4k\arctan(\sqrt{t}) \leq 4k\sqrt{t} \leq 4k\sqrt{\frac{1}{2k(2k-1)}}$, and the latter quantity is at most~$\sqrt{2}$ for all $k \geq 2$.  Since $\cos(x)$ is decreasing for $x \in [0, \sqrt{2}]$ we have
    \begin{multline*}
        \cos(4k\arctan(\sqrt t)) \geq\cos(4k \sqrt{t})\geq \kappa(4k\sqrt{t}) \\
    \begin{aligned}
        \Rightarrow\quad   Q_{k,0}(t) &\geq {-1} + (1+2kt)\left(1 - \tfrac{1}{4k} + \tfrac{1}{4k} \kappa(4k\sqrt{t})\right)\\
        &=    -1 + (1 + 2kt) \left(1 - \tfrac{1}{4k} + \tfrac{1}{4k} \left(1 - 8k^2t + \tfrac{32}{3} k^4t^2 - \tfrac{256}{45}k^6t^3\right) \right)\\
        &= q(t) t^2, \qquad \text{where } q(t) = -\tfrac{128}{45}k^6 t^2 - \left(\tfrac{64}{45}k - \tfrac{16}{3}\right)k^4 t + \left(\tfrac{8}{3}k-4\right)k^2.
    \end{aligned}
    \end{multline*}
    It remains to show that $q(t) \geq 0$ for $0 \leq t \leq \frac{1}{2k(2k-1)}$. Since $q(t)$ is a quadratic polynomial with negative leading coefficient, we only need to check that $q(0),~q(\frac{1}{2k(2k-1)}) \geq 0$.
    We have $q(0) = \left(\tfrac{8}{3}k-4\right)k^2$, which is clearly nonnegative for $k \geq 2$.  Finally, one may check that
    \begin{align*}
        q\left(\tfrac{1}{2k(2k-1)}\right) 
        &= \tfrac{4k^2}{45(2k-1)^2}\left(19 + 136(k-2)\left((k-\tfrac{13}{34})^2+\tfrac{103}{1156}\right)\right),
    \end{align*}
    which is evidently nonnegative for $k \geq 2$.
\end{proof}

\begin{proposition}               \label{prop:mathoverflow}
    For all  $1 \leq i \leq k \in \N^+$,  the polynomial $Q_{k,i}(t)$ is nonnegative  (with $\rho^* = 1-\frac{1}{2k}$).
\end{proposition}
\begin{proof}
In fact, we will prove the stronger claim that each $Q_{k,i}(t)$ is nonnegative even when~$\rho^*$ is set to~$0$.  I.e., we will show that
\[
    \wt{Q}_{k, i}(t) \coloneqq \tfrac{1}{2}\tbinom{2k}{2i}(1+t)^{2k-2i} + \tfrac{1}{2} \sum_{j=i}^k \tbinom{2k}{2j} (1-t)^{2k-2j} \tbinom{j}{i}(-4t)^{j-i}
\]
is nonnegative. To see that this is indeed stronger, simply note that $Q_{k,i}(t)$ and $\wt{Q}_{k,i}(t)$ are convex combinations of the same two main quantities, but $\wt{Q}_{k,i}(t)$ has less of its ``weight'' on the first quantity $\tbinom{2k}{2i}(1+t)^{2k-2i}$, which is clearly nonnegative. We will furthermore show that even $\wt{Q}_{k,0}(t) \geq 0$.

This is not particularly easy to prove by hand, but using computer assistance yields a compact proof.
Using automated guessing~\cite{kauers09a} one can discover that the following recurrence seems to hold for all integers $0\leq i\leq k$:
\[
    (1+i)(1+k)\wt{Q}_{k+2, i+1}(t) = (1+i)(2+k)(1+t)^2 \wt{Q}_{k+1, i+1}(t) + (2+k)(2+2k-i) \wt{Q}_{k+1,i}(t).
\]
The correctness of this recurrence was shown by computing an ideal of annihilating operators for $\wt{Q}_{k,i}(t)$
from the sum definition using creative telescoping and holonomic closure properties, and then showing by a Gr\"obner
basis computation that the guessed recurrence belongs to this ideal. A detailed description of this calculation
would lead a bit to far, the interested reader is referred to \cite{koutschan13,kauers13,kauers14c} for recent introductions
to the relevant computational techniques and the algebraic theory behind them.

In light of the guessed-and-proved recurrence for $\wt Q_{k,i}(t)$ we only need to prove $\tilde{Q}_{k, i}(t) \geq 0$ for the cases that $k = i$ and $i = 0$; the nonnegativity of $\wt{Q}_{k, i}(t)$ for general~$k$ and~$i$ then follows by induction. For $k= i$ we have $\wt{Q}_{k, k}(t) = 1 \geq 0$. For $i= 0$ the proof of nonnegativity is similar to, but easier than, that of Proposition~\ref{prop:Q0}. For $t < 0$ it's obvious from its definition that $\tilde{Q}_{k, 0}(t)$ is nonnegative. For $t \geq 0$, the proof of Proposition~\ref{prop:Q0} gives
\[
    \wt{Q}_{k, 0}(t) = \tfrac{1}{2} (1+t)^{2k}\left(1 +  \cos(4k \arctan(\sqrt{t}))\right) \geq 0. \qedhere
\]
\end{proof}

\section{The Frankl-R\"odl Theorem in SOS}         \label{sec:BHM}

The applications of our work to $3$-Coloring and Vertex-Cover follow by giving a low-degree SOS proof of the Frankl--R\"odl Theorem, that $\maxis(\FR{n}{\gamma}) < o(1)$.  (See Section~\ref{sec:intro-application} for the definition of the Frankl--R\"odl graphs and the statement of the Frankl--R\"odl Theorem.)  More precisely, in this section we will prove the following:
\begin{theorem}                                     \label{thm:main}
    Let $n \in \N^+$ and let $\frac{1}{\log n} \leq \gamma \leq \frac14$ be such that $(1-\gamma)n$ is an even integer.  Given the Frankl--R\"odl graph $\FR{n}{\gamma} = (V,E)$, for each $x \in V = \bn$ let $f(x)$ be an indeterminate. Then there is a degree-$\left(4 \lceil \frac{1}{4\gamma}\rceil\right)$  SOS refutation of $\maxis(G) \geq \Omega(n^{-\gamma/10})$; i.e.,
    \begin{multline*}
        \{ f(x)^2 = f(x)\ \forall x \in V, \ \ f(x) f(y) = 0\ \forall (x, y) \in E, \ \  \tfrac{1}{|V|} \littlesum_{x \in V} f(x) \geq C n^{-\gamma/10} \}
         \ \ \proves{4 \lceil \frac{1}{4\gamma}\rceil}\ \ -1 \geq 0
    \end{multline*}
    for a universal constant $C$.
\end{theorem}

In particular, this theorem shows that the degree-$4$ SOS/Lasserre algorithm certifies that $\maxis(\FR{n}{1/4}) < o(1)$, which is a stronger statement than the chromatic number bound $\chi(\FR{n}{1/4}) = \omega(1)$.  More generally, it shows that the $N^{O(1/\gamma)}$-time SOS/Lasserre hierarchy algorithm certifies that $\maxis(\FR{n}{\gamma}) \leq O(n^{-\gamma/10})$; in other words, that $\minvc(\FR{n}{\gamma}) \geq (1-O(n^{-\gamma/10}))N = (1-O(n^{-\gamma/10})) 2^n$.  Note that this bound is nontrivial only for $\gamma \gg \frac{1}{\log n}$; the reason for this dependence on $\gamma$ will be seen shortly.

\medskip

We prove Theorem~\ref{thm:main} by ``SOS-izing'' the Benabbas--Hatami--Magen Fourier-theoretic proof~\cite{BHM12} of the following ``density'' version of the Frankl--R\"odl Theorem:
\begin{theorem}                                     \label{thm:bhm}
    (\cite{BHM12})  Fix $0 < \gamma < 1/2$ and $0 < \alpha \leq 1$.  In the graph $\FR{n}{\gamma} = (V,E)$, if $S \subseteq V$ has $|S|/2^n \geq \alpha$ then
    \[
        \Pr_{(\bx, \by) \sim E}[\bx \in S, \by \in S] \geq 2(\alpha/2)^{1/\gamma} - o_n(1).
    \]
\end{theorem}
Here the $o_n(1)$ goes to $0$ rather slowly in~$n$, which means that the Benabbas--Hatami--Magen proof only recovers the Frankl--R\"odl Theorem for $\gamma > \omega(\frac{1}{\log n})$.  This is due to comparison between the $T_d$ and $S_d$ operators described below; it seems possible that some additional technical work would allow for smaller values of~$\gamma$.

\subsection{The Benabbas--Hatami--Magen argument in SOS}

Benabbas, Hatami, and Magen~\cite{BHM12} introduce the following operator:
\begin{definition}
    For integer $0 \leq d \leq n$ the operator $S_d$ is defined on functions $f \btR$ by $S_d f(x) = \Ex_{\by}[f(\by)]$,
    where $\by$ is chosen uniformly at random subject to $\hamdist(x,\by) = d$.
\end{definition}

The key technical contribution of~\cite{BHM12} is showing how to pass between the $S_d$ operators (which are relevant for Frankl--R\"odl analysis) and the $T_\rho$ operators (for which we have reverse hypercontractivity).  Intuitively, the operators $S_d$ and $T_{1-2d/n}$ should be similar (at least if  $d/n$ is bounded away from~$0$ and~$1$). However there is one caveat: ``parity'' issues with $S_d$.  For example, if $f \co \bn \to \{0,1\}$ is the indicator of the strings of even Hamming weight, then
\[
    \la f, S_d f \ra = \begin{cases} 0 & \text{if $d$ is odd,} \\ \half & \text{if $d$ is even;} \end{cases} \qquad \text{but, } \la f, T_{1-2d/n} f\ra \approx \tfrac14 \text{ for $d$ odd or even}.
\]
Benabbas, Hatami, and Magen evade this parity issue by considering the operator $\half S_{d} + \half S_{d+1}$.
\begin{definition}
    For integer $0 \leq d < n$, we define the operator $S_d' = \half S_{d} + \half S_{d+1}$.
\end{definition}
The crucial theorem in~\cite{BHM12}'s work is the following:\rnote{does it really follow?  a lot of small details to double-check}
\begin{theorem}                                     \label{thm:op-norm}
    (Follows from Lemma 3.4 in~\cite{BHM12}.)  Let $f \btR$.  Let $d = n - c$ for some integer $e^2\sqrt{n} \leq c \leq n/2$.  Then for $\rho = 1-2d/n$,
    \[
        \la f, S'_d f \ra - \la f, T_\rho f \ra = \sum_{U \subseteq [n]} \wh{f}(U)^2 \cdot \delta(U),
    \]
    where each real number $\delta(U)$ satisfies
    \[
        |\delta(U)| \leq O(\max\{n^{-1/5}, \tfrac{n}{c^2}\log^2(\tfrac{c^2}{n})\}).
    \]
\end{theorem}

Given  Theorem~\ref{thm:op-norm}, Benabbas, Hatami, and Magen are able to deduce their main Theorem~\ref{thm:bhm} from the reverse hypercontractivity result Theorem~\ref{thm:morss} without too much trouble.  We now show that this deduction can also be carried out in the SOS proof system.  Specifically, we give here the proof of our Theorem~\ref{thm:main}, relying on the SOS proof of hypercontractivity (Theorem~\ref{thm:our-SOS-reverse}) from Section~\ref{sec:reverse-hypercon}.

\begin{proof}[Proof of Theorem~\ref{thm:main}]
    Write $d = (1-\gamma)n$ (where $\frac{1}{\log n} \leq \gamma \leq \frac14$) and write $\rho' = 1 - 2d/n = -(1-2\gamma)$. For $i = 0, 1$ let us denote
    \[
        f_i(x) = \begin{cases}
                    f(x) & \text{if $x$'s Hamming weight equals $i$ mod 2,} \\ 0 & \text{else.}
                 \end{cases}
    \]
    We have
    \[
        \{ f(x) f(y) = 0\ \forall \hamdist(x, y) = d \} \quad\proves{2}\quad \la f_0, S'_d f_0 \ra + \la f_1, S'_d f_1 \ra = 0
    \]
    because if $x$'s Hamming weight has the same parity as $y$'s then their distance can only be~$d$ (an even integer) not $d+1$ (an odd one).\rnote{technically, am using some basic unspoken SOS maneuvers re equality here}  Using Theorem~\ref{thm:op-norm} it follows that
    \begin{multline*}
         \{ f(x) f(y) = 0\ \forall \hamdist(x, y) = d \} \quad\proves{2} \\
         \la f_0, T_{\rho'} f_0 \ra + \la f_1, T_{\rho'} f_1 \ra \leq \delta \littlesum_{U} \wh{f_0}(U)^2 +  \delta \littlesum_{U} \wh{f_1}(U)^2 = \delta (\E[f_0^2] +\E[f_1^2]) = \delta \E[f^2],
    \end{multline*}
    where
    \[
        \delta = O(\max\{n^{-1/5}, \tfrac{1}{\gamma^2 n}\log^2(\gamma^2 n)\}) = O(n^{-1/5})
    \]
    (with the second bound using $\gamma \geq \frac{1}{\log n}$.)    We now write $g_i(x)$ to denote $f_i(-x)$ and also $\rho = -\rho' = 1-2\gamma$; then
    \[
        \la f_i, T_{\rho'} f_i \ra = \la f_i, T_{\rho} g_i \ra = \Ex_{\substack{(\bx,\by)\\ \rho\text{-corr'd}}}[f_i(\bx) g_i(\by)]
    \]
    so we conclude
    \[
         \{ f(x) f(y) = 0\ \forall \hamdist(x, y) = d \} \quad\proves{2}\quad
         \Ex_{\substack{(\bx,\by)\\ \rho\text{-corr'd}}}[f_0(\bx)g_0(\by)]  + \Ex_{\substack{(\bx,\by)\\ \rho\text{-corr'd}}}[f_1(\bx)g_1(\by)] \leq \delta \E[f^2].
    \]
    Next, define
    \[
        k = \left\lceil \tfrac1{4\gamma} \right\rceil \geq 1.
    \]
    We have $f(x)^2 = f(x) \proves{2k} f(x)^{2k} = f(x)$, from which we may easily deduce
    \begin{multline*}
         \{ f(x)^2 = f(x)\ \forall x, \quad f(x) f(y) = 0\ \forall \hamdist(x, y) = d \} \\
         \proves{4k}\quad
         \Ex_{\substack{(\bx,\by)\\ \rho\text{-corr'd}}}[f_0(\bx)^{2k}g_0(\by)^{2k}]  + \Ex_{\substack{(\bx,\by)\\ \rho\text{-corr'd}}}[f_1(\bx)^{2k}g_1(\by)^{2k}] \leq \delta \E[f].
    \end{multline*}
    Since $\rho =1-2\gamma \leq 1-\frac{1}{2k}$ we may apply our reverse hypercontractivity result Theorem~\ref{thm:our-SOS-reverse} (in Section~\ref{sec:reverse-hypercon}) to deduce
    \begin{multline*}
         \{ f(x)^2 = f(x)\ \forall x, \quad f(x) f(y) = 0\ \forall \hamdist(x, y) = d \} \quad       \proves{4k}\quad
         \E[f_0]^{2k}\E[g_0]^{2k} + \E[f_1]^{2k}\E[g_1]^{2k} \leq \delta \E[f].
    \end{multline*}
    We're now almost done.  First, $\E[f_i] = \E[g_i]$ formally for
    each $i = 0, 1$.  Second, for simplicity we use the bound
    \[
        \{ f(x)^2 = f(x)\ \forall x\} \quad\proves{2}{\quad} \delta \E[f] = \delta \E[2f-f^2] = \delta \E[1 -(1-f)^2] \leq \delta.
    \]
    Thus we have
    \begin{align*}
         \{ f(x)^2 = f(x)\ \forall x, \quad f(x) f(y) = 0\ \forall \hamdist(x, y) = d \} \quad
         \proves{4k}\quad
         \delta &\geq \E[f_0]^{4k} + \E[f_1]^{4k} \\
                &\geq 2\left(\frac{\E[f_0] + \E[f_1]}{2}\right)^{4k} \\
                &= 2(\E[f]/2)^{4k} \\
         \Rightarrow \quad &2^{4k-1} \delta \geq \E[f]^{4k},
    \end{align*}
    where the second inequality is Lemma~\ref{lem:super-CS}.  Finally, from Lemma~\ref{lem:powers} we may deduce
    \[
        \E[f] \geq C n^{-\gamma/10} \quad \proves{4k} \quad  \E[f]^{4k} \geq C^{4k} n^{-4k\gamma/10} = C^{4k} n^{-(2/5)\lceil\frac{1}{4\gamma}\rceil\gamma} \geq C^{4k} n^{-1/5} \geq 2^{4k} \delta
    \]
    for $C$ sufficiently large, using our upper bound on~$\delta$. Combining the previous two statements we can get
    \begin{multline*}
         \{ f(x)^2 = f(x)\ \forall x, \quad f(x) f(y) = 0\ \forall \hamdist(x, y) = d, \quad \E[f] \geq Cn^{-\gamma/10} \}
         \quad \proves{4k} \quad -1 \geq 0,
    \end{multline*}
    as required.
\end{proof}

\section{Conclusions}   \label{sec:conclusions}
We describe here a few questions left open by our work.  
Regarding reverse hypercontractivity, it seems we may not have given the Book proof of the SOS Two-Point Inequality. We would be happy to see a more elegant ``human proof'', but even more interesting would be a computer algebra technique that could automatically prove SOS-ness, symbolically for all~$k$.

An additional open question regarding the Frankl--R\"odl Theorem is whether the Benabbas--Hatami--Magen proof can be improved to work for~$\gamma$ as small as~$\sqrt{\frac{\log n}{n}}$.   However even if this is possible, the resulting SOS proof would (presumably) be of degree $\Omega(\sqrt{\frac{n}{\log n}}) = \Omega(\sqrt{\frac{\log N}{\log \log N}})$, slightly superconstant.  Could there be an $O(1)$-degree SOS of the Frankl--R\"odl Theorem with this setting of~$\gamma$, or should one try to prove an SOS lower bound?  An interesting toy version of this question is the following:  The vertex isoperimetric inequality for the hypercube immediately implies that if $A, B \subseteq \bn$ satisfy $\dist(A,B) \geq \sqrt{n \log n}$ then $\frac{|A|}{2^n}\frac{|B|}{2^n} = o_n(1)$.  Does this have an $O(1)$-degree SOS proof?

\subsection*{Acknowledgments}

The authors would like to thank Siavosh Benabbas and Hamed Hatami for the advance copy of~\cite{BHM12}. Proposition~\ref{prop:mathoverflow} was independently proven by Fedor Nazarov; we are very grateful to him for sharing his proof with us.   Thanks also to Rajsekar Manokaran, Toni Pitassi, and Doron Zeilberger for helpful discussions.

\bibliographystyle{alpha}
\bibliography{../bib/odonnell-bib}

\newcommand{\etalchar}[1]{$^{#1}$}
\begin{thebibliography}{MOR{\etalchar{+}}06}

\bibitem[ABL02]{ABL02}
Sanjeev Arora, B{\'e}la Bollob{\'a}s, and L{\'a}szl{\'o} Lov{\'a}sz.
\newblock Proving integrality gaps without knowing the linear program.
\newblock In {\em Proceedings of the 43rd Annual IEEE Symposium on Foundations
  of Computer Science}, pages 313--322, 2002.

\bibitem[ABLT06]{ABLT06}
Sanjeev Arora, B{\'e}la Bollob{\'a}s, L{\'a}szl{\'o} Lov{\'a}sz, and Iannis
  Tourlakis.
\newblock Proving integrality gaps without knowing the linear program.
\newblock {\em Theory of Computing}, 2(1):19--51, 2006.

\bibitem[ABS10]{ABS10}
Sanjeev Arora, Boaz Barak, and David Steurer.
\newblock Subexponential algorithms for {U}nique {G}ames and related problems.
\newblock In {\em Proceedings of the 51st Annual IEEE Symposium on Foundations
  of Computer Science}, pages 563--572, 2010.

\bibitem[ACC06]{ACC06}
Sanjeev Arora, Moses Charikar, and Eden Chlamtac.
\newblock New approximation guarantee for chromatic number.
\newblock In {\em Proceedings of the 38th Annual ACM Symposium on Theory of
  Computing}, pages 215--224, 2006.

\bibitem[AG11]{AG11}
Sanjeev Arora and Rong Ge.
\newblock New tools for graph coloring.
\newblock In {\em Proceedings of the 14th Annual International Workshop on
  Approximation Algorithms for Combinatorial Optimization Problems}, pages
  1--12, 2011.

\bibitem[AK98]{AK98a}
Noga Alon and Nabil Kahale.
\newblock Approximating the independence number via the {$\vartheta$}-function.
\newblock {\em Mathematical Programming}, 80(3, Ser. A):253--264, 1998.

\bibitem[BBH{\etalchar{+}}12]{BBH+12}
Boaz Barak, Fernando Brand{\~a}o, Aram Harrow, Jonathan Kelner, David Steurer,
  and Yuan Zhou.
\newblock Hypercontractivity, sum-of-squares proofs, and their applications.
\newblock In {\em Proceedings of the 44th Annual ACM Symposium on Theory of
  Computing}, pages 307--326, 2012.

\bibitem[BCGM11]{BCGM11}
Siavosh Benabbas, Siu~On Chan, Konstantinos Georgiou, and Avner Magen.
\newblock Tight gaps for {V}ertex {C}over in the {S}herali--{A}dams {SDP}
  hierarchy.
\newblock In {\em Proceedings of the 32nd Annual IARCS Conference on
  Foundations of Software Technology and Theoretical Computer Science}, pages
  41--54, 2011.

\bibitem[BHM12]{BHM12}
Siavosh Benabbas, Hamed Hatami, and Avner Magen.
\newblock An isoperimetric inequality for the {H}amming cube with applications
  for integrality gaps in degree-bounded graphs.
\newblock Unpublished, 2012.

\bibitem[Bon70]{Bon70}
Aline Bonami.
\newblock {\'E}tude des coefficients {F}ourier des fonctions de {$L^{p}(G)$}.
\newblock {\em Annales de l'Institute Fourier}, 20(2):335--402, 1970.

\bibitem[Bor82]{Bor82}
Christer Borell.
\newblock Positivity improving operators and hypercontractivity.
\newblock {\em Mathematische Zeitschrift}, 180:225--234, 1982.

\bibitem[Cha02]{Cha02}
Moses Charikar.
\newblock On semidefinite programming relaxations for graph coloring and vertex
  cover.
\newblock In {\em Proceedings of the 13th Annual ACM-SIAM Symposium on Discrete
  Algorithms}, pages 616--620, 2002.

\bibitem[CMM09]{CMM09}
Moses Charikar, Konstantin Makarychev, and Yury Makarychev.
\newblock Integrality gaps for {S}herali--{A}dams relaxations.
\newblock In {\em Proceedings of the 41st Annual ACM Symposium on Theory of
  Computing}, pages 283--292. ACM, 2009.

\bibitem[DKSV06]{DKSV06}
Nikhil Devanur, Subhash Khot, Rishi Saket, and Nisheeth Vishnoi.
\newblock Integrality gaps for {S}parsest {C}ut and {M}inimum {L}inear
  {A}rrangement problems.
\newblock In {\em Proceedings of the 38th Annual ACM Symposium on Theory of
  Computing}, pages 537--546, 2006.

\bibitem[DMN13]{DMN13}
Anindya De, Elchanan Mossel, and Joe Neeman.
\newblock Majority is {S}tablest : {D}iscrete and {SoS}.
\newblock In {\em Proceedings of the 45th Annual ACM Symposium on Theory of
  Computing}, 2013.

\bibitem[DMR09]{DMR09}
Irit Dinur, Elchanan Mossel, and Oded Regev.
\newblock Conditional hardness for approximate coloring.
\newblock {\em SIAM Journal on Computing}, 39(3):843--873, 2009.

\bibitem[DS05]{DS05}
Irit Dinur and Samuel Safra.
\newblock On the hardness of approximating minimum vertex cover.
\newblock {\em Annals of Mathematics}, 162(1):439--485, 2005.

\bibitem[FKO07]{FKO07}
Uriel Feige, Guy Kindler, and Ryan O'Donnell.
\newblock Understanding parallel repetition requires understanding foams.
\newblock In {\em Proceedings of the 22nd Annual IEEE Conference on
  Computational Complexity}, pages 179--192, 2007.

\bibitem[FLS04]{FLS04}
Uriel Feige, Michael Langberg, and Gideon Schechtman.
\newblock Graphs with tiny vector chromatic numbers and huge chromatic numbers.
\newblock {\em SIAM Journal on Computing}, 33(6):1338--1368, 2004.

\bibitem[FO06]{FO06}
Uriel Feige and Eran Ofek.
\newblock Random 3{CNF} formulas elude the {L}ov{\'a}sz theta function.
\newblock Technical Report cs/0603084, arXiv, 2006.

\bibitem[FR87]{FR87}
P\'eter Frankl and Vojt{\v e}ch R{\"o}dl.
\newblock Forbidden intersections.
\newblock {\em Transactions of the American Mathematical Society},
  300(1):259--286, 1987.

\bibitem[Geo10]{Geo10}
Konstantinos Georgiou.
\newblock {\em Integrality gaps for strong linear programming and semidefinite
  programming relaxations}.
\newblock PhD thesis, University of Toronto, 2010.

\bibitem[GJ79]{GJ79}
Michael Garey and David Johnson.
\newblock {\em Computers and Intractability: a guide to the theory of
  {NP}-completeness}.
\newblock W. H. Freeman and Company, 1979.

\bibitem[GK04]{GK04}
Venkatesan Guruswami and Sanjeev Khanna.
\newblock On the hardness of 4-coloring a 3-colorable graph.
\newblock {\em SIAM Journal on Discrete Mathematics}, 18(1):30--40, 2004.

\bibitem[GM08]{GM08}
Konstantinos Georgiou and Avner Magen.
\newblock Expansion fools the {S}herali--{A}dams system: compromising local and
  global arguments.
\newblock Technical Report CSRG--587, University of Toronto, November 2008.

\bibitem[GM10]{GM10}
Konstantinos Georgiou and Avner Magen.
\newblock Tight integrality gap for {S}herali--{A}dams {SDPs} for {V}ertex
  {C}over.
\newblock Manuscript, 2010.

\bibitem[GMPT10]{GMPT10}
Konstantinos Georgiou, Avner Magen, Toniann Pitassi, and Iannis Tourlakis.
\newblock Integrality gaps of $2-o(1)$ for {V}ertex {C}over {SDPs} in the
  {L}ov{\'a}sz--{S}chrijver hierarchy.
\newblock {\em SIAM Journal on Computing}, 39(8):3553--3570, 2010.

\bibitem[GMT08]{GMT08}
Konstantinos Georgiou, Avner Magen, and Iannis Tourlakis.
\newblock Vertex {C}over resists {SDPs} tightened by local hypermetric
  inequalities.
\newblock In {\em Proceedings of the 12th Annual Conference on Integer
  Programming and Combinatorial Optimization}, IPCO'08, pages 140--153, Berlin,
  Heidelberg, 2008. Springer-Verlag.

\bibitem[GV01]{GV01}
Dima Grigoriev and Nicolai Vorobjov.
\newblock Complexity of {N}ull- and {P}ositivstellensatz proofs.
\newblock {\em Annals of Pure and Applied Logic}, 113(1):153--160, 2001.

\bibitem[Kau09]{kauers09a}
  Manuel Kauers.
  \newblock \emph{Guessing Handbook.}
  \newblock RISC-Linz Technical Report 09-07.
  \newblock 2009.

\bibitem[Kau13]{kauers13}
  Manuel Kauers.
  \newblock \emph{The Holonomic Toolkit.}
  \newblock In Johannes Bl\"umlein and Carsten Schneider (editors), Computer Algebra in Quantum Field Theory: Integration, Summation and Special Functions.
  \newblock Springer.
  \newblock 2013.
  
\bibitem[Kau14]{kauers14c}
  Manuel Kauers.
  \newblock \emph{Computer Algebra.}
  \newblock In Miklos Bona (editor), Handbook of Enumerative Combinatorics.
  \newblock Taylor and Francis.
  \newblock 2015.

\bibitem[Kou13]{koutschan13}
  Christoph Koutschan
  \newblock \emph{Creative Telescoping for holonomic functions.}
  \newblock In Johannes Bl\"umlein and Carsten Schneider (editors), Computer Algebra in Quantum Field Theory: Integration, Summation and Special Functions.
  \newblock Springer.
  \newblock 2013.

\bibitem[Kel12]{Kel12}
Nathan Keller.
\newblock A tight quantitative version of {A}rrow's impossibility theorem.
\newblock {\em Journal of the European Mathematical Society (JEMS)},
  14(5):1331--1355, 2012.

\bibitem[KG98]{KG98}
Jon Kleinberg and Michel Goemans.
\newblock The {L}ov\'asz theta function and a semidefinite programming
  relaxation of vertex cover.
\newblock {\em SIAM Journal on Discrete Mathematics}, 11(2):196--204, 1998.

\bibitem[KKL88]{KKL88}
Jeff Kahn, Gil Kalai, and Nathan Linial.
\newblock The influence of variables on {B}oolean functions.
\newblock In {\em Proceedings of the 29th Annual IEEE Symposium on Foundations
  of Computer Science}, pages 68--80, 1988.

\bibitem[KLS00]{KLS00}
Sanjeev Khanna, Nathan Linial, and Shmuel Safra.
\newblock On the hardness of approximating the chromatic number.
\newblock {\em Combinatorica}, 20(3):393--415, 2000.

\bibitem[KMS98]{KMS98}
David Karger, Rajeev Motwani, and Madhu Sudan.
\newblock Approximate graph coloring by semidefinite programming.
\newblock {\em Journal of the ACM}, 45(2):246--265, 1998.

\bibitem[KR08]{KR08}
Subhash Khot and Oded Regev.
\newblock Vertex {C}over might be hard to approximate to within $2-\epsilon$.
\newblock {\em Journal of Computer and System Sciences}, 74(3):335--349, 2008.

\bibitem[KS88]{KS88}
Wieslaw Krakowiak and Jerzy Szulga.
\newblock Hypercontraction principle and random multilinear forms.
\newblock {\em Probability Theory and Related Fields}, 77(3):325--342, 1988.

\bibitem[KS09]{KS09}
Subhash Khot and Rishi Saket.
\newblock {SDP} integrality gaps with local $\ell_1$-embeddability.
\newblock In {\em Proceedings of the 50th Annual IEEE Symposium on Foundations
  of Computer Science}, pages 565--574, 2009.

\bibitem[KV05]{KV05}
Subhash Khot and Nisheeth Vishnoi.
\newblock The {Unique Games Conjecture}, integrality gap for cut problems and
  embeddability of negative type metrics into $\ell_1$.
\newblock In {\em Proceedings of the 46th Annual IEEE Symposium on Foundations
  of Computer Science}, pages 53--62, 2005.

\bibitem[Las00]{Las00}
Jean Lasserre.
\newblock Optimisation globale et th{\'e}orie des moments.
\newblock {\em Comptes Rendus de l'Acad\'{e}mie des Sciences},
  331(11):929--934, 2000.

\bibitem[Lov79]{Lov79a}
L{\'a}szl{\'o} Lov{\'a}sz.
\newblock On the {S}hannon capacity of a graph.
\newblock {\em Transactions on Information Theory}, 25(1):1--7, 1979.

\bibitem[MOO05]{MOO05}
Elchanan Mossel, Ryan O'Donnell, and Krzysztof Oleszkiewicz.
\newblock Noise stability of functions with low influences: invariance and
  optimality.
\newblock In {\em Proceedings of the 46th Annual IEEE Symposium on Foundations
  of Computer Science}, pages 21--30, 2005.

\bibitem[MOO10]{MOO10}
Elchanan Mossel, Ryan O'Donnell, and Krzysztof Oleszkiewicz.
\newblock Noise stability of functions with low influences: invariance and
  optimality.
\newblock {\em Annals of Mathematics}, 171(1), 2010.

\bibitem[MOR{\etalchar{+}}06]{MOR+06}
Elchanan Mossel, Ryan O'Donnell, Oded Regev, Jeffrey Steif, and Benjamin
  Sudakov.
\newblock Non-interactive correlation distillation, inhomogeneous {Markov}
  chains, and the reverse {Bonami--Beckner} inequality.
\newblock {\em Israel Journal of Mathematics}, 154:299--336, 2006.

\bibitem[Mos12a]{Mos12a}
Elchanan Mossel.
\newblock A quantitative {A}rrow theorem.
\newblock {\em Probability Theory and Related Fields}, 154(1-2):49--88, 2012.

\bibitem[MOS12b]{MOS12b}
Elchanan Mossel, Krzysztof Oleszkiewicz, and Arnab Sen.
\newblock On reverse hypercontractivity.
\newblock Technical Report 1108.1210, arXiv, 2012.

\bibitem[MR12]{MR12}
Elchanan Mossel and Mikl{\'o}s R{\'a}cz.
\newblock A quantitative {G}ibbard--{S}atterthwaite theorem without neutrality.
\newblock In {\em Proceedings of the 53rd Annual IEEE Symposium on Foundations
  of Computer Science}, pages 1041--1060, 2012.

\bibitem[OZ13]{OZ13}
Ryan O'Donnell and Yuan Zhou.
\newblock Approximability and proof complexity.
\newblock In {\em Proceedings of the 24th Annual ACM-SIAM Symposium on Discrete
  Algorithms}, 2013.

\bibitem[Par00]{Par00}
Pablo Parrilo.
\newblock {\em Structured Semidefinite Programs and Semialgebraic Geometry
  Methods in Robustness and Optimization}.
\newblock PhD thesis, California Institute of Technology, 2000.

\bibitem[PWZ97]{PWZ97}
Marko Petkov{\v s}ek, Herbert Wilf, and Doron Zeilberger.
\newblock {\em $A=B$}.
\newblock AK Peters, Ltd., 1997.

\bibitem[RS09]{RS09b}
Prasad Raghavendra and David Steurer.
\newblock Integrality gaps for strong {SDP} relaxations of {U}nique {G}ames.
\newblock In {\em Proceedings of the 50th Annual IEEE Symposium on Foundations
  of Computer Science}, pages 575--585, 2009.

\bibitem[Sch08]{Sch08}
Grant Schoenebeck.
\newblock Linear level {L}asserre lower bounds for certain $k$-{CSP}s.
\newblock In {\em Proceedings of the 49th Annual IEEE Symposium on Foundations
  of Computer Science}, pages 593--602, 2008.

\bibitem[She09]{She09a}
Jonah Sherman.
\newblock Breaking the multicommodity flow barrier for {$O(\sqrt{\log
  n})$}-approximations to {S}parsest {C}ut.
\newblock In {\em Proceedings of the 50th Annual IEEE Symposium on Foundations
  of Computer Science}, pages 363--372, 2009.

\bibitem[STT07a]{STT07a}
Grant Schoenebeck, Luca Trevisan, and Madhur Tulsiani.
\newblock A linear round lower bound for {L}ov{\'a}sz--{S}chrijver {SDP}
  relaxations of {V}ertex {C}over.
\newblock In {\em Proceedings of the 22nd Annual IEEE Conference on
  Computational Complexity}, pages 205--216, 2007.

\bibitem[STT07b]{STT07b}
Grant Schoenebeck, Luca Trevisan, and Madhur Tulsiani.
\newblock Tight integrality gaps for {L}ov{\'a}sz--{S}chrijver {LP} relaxations
  of {V}ertex {C}over and {M}ax {C}ut.
\newblock In {\em Proceedings of the 39th Annual ACM Symposium on Theory of
  Computing}, pages 302--310. ACM, 2007.

\bibitem[Tou06]{Tou06}
Iannis Tourlakis.
\newblock New lower bounds for {V}ertex {C}over in the lov{\'a}sz--{S}chrijver
  hierarchy.
\newblock In {\em Proceedings of the 21st Annual IEEE Conference on
  Computational Complexity}, pages 170--182. IEEE Computer Society, 2006.

\bibitem[Tul09]{Tul09}
Madhur Tulsiani.
\newblock {CSP} gaps and reductions in the {L}asserre hierarchy.
\newblock In {\em Proceedings of the 41st Annual ACM Symposium on Theory of
  Computing}, pages 303--312, 2009.

\bibitem[Zei90]{Zei90}
Doron Zeilberger.
\newblock A fast algorithm for proving terminating hypergeometric identities.
\newblock {\em Discrete Mathematics}, 80(2):207--211, 1990.

\end{thebibliography}

\appendix

\section{Solution to the puzzle}
\begin{multline*}
{a}^{6}{b}^{6}+15 \left( a+b \right) ^{2}
 \left( 1+ba \right) ^{4}+10 \left( a+b
 \right) ^{4} \left( 1+ba \right) ^{2}+5 \left( {
a}^{3}b+{b}^{2}+{a}^{2}+a{b}^{3}
 \right) ^{2}\\+\ 35{a}^{4}{b}^{4}+ \left( {a}^{3}+{
b}^{3} \right) ^{2}+{\tfrac {17}{3}} \left( {a}^{2}b
+a{b}^{2} \right) ^{2}+35{a}^{2}{b}^{2}+{
\tfrac {148}{3}}{a}^{2}{b}^{4}+{\tfrac {148}{3}}{
a}^{4}{b}^{2}.
\end{multline*}

\end{document}